\documentclass{article}
\usepackage[left=2.5cm, top=3cm, bottom=3cm, right=2.5cm]{geometry}
\usepackage{amsfonts}
\usepackage{amsthm}
\usepackage{amsmath}
\usepackage{amssymb}
\usepackage{enumerate}
\usepackage[numbers,sort&compress]{natbib}

\everymath{\displaystyle}

\newcommand{\dd}{\mathrm{d}}

\newcommand{\ud}{\,\dd}
\newcommand{\nbr}[1]{$#1$\nobreakdash-\hspace{0pt}}
\newcommand{\braket}[1]{\langle{#1}\rangle}

\newcommand{\lshad}{[\![}
\newcommand{\rshad}{]\!]}
\newcommand{\sdot}{\,\cdot\,}
\renewcommand{\vec}[1]{\boldsymbol{\mathbf{#1}}}
\providecommand{\abs}[1]{\lvert#1\rvert}
\providecommand{\Abs}[1]{\left\lvert#1\right\rvert}
\providecommand{\norm}[1]{\lVert#1\rVert}

\DeclareMathOperator{\id}{id}
\DeclareMathOperator{\tr}{tr}
\DeclareMathOperator{\Tr}{Tr}

\DeclareMathOperator{\im}{Im}
\DeclareMathOperator{\ad}{ad}
\DeclareMathOperator{\Ad}{Ad}
\DeclareMathOperator{\T}{T}
\DeclareMathOperator{\supp}{supp}
\DeclareMathOperator{\Spec}{Spec}

\newtheorem{theorem}{Theorem}

\newtheorem{lemma}{Lemma}
\newtheorem{corollary}{Corollary}
\theoremstyle{definition}

\title{Deformation quantization on the cotangent bundle of a Lie group}
\author{Ziemowit Doma\'nski\thanks{The author has been partially supported by
the grant 04/43/DSPB/0094 from the Polish Ministry of Science and Higher
Education.}\\
\small Institute of Mathematics, Pozna{\'n} University of Technology\\
\small Piotrowo 3A, 60-965 Pozna{\'n}, Poland\\
\small \tt ziemowit.domanski@put.poznan.pl}

\begin{document}
\maketitle

\begin{abstract}
We develop a complete theory of non-formal deformation quantization on the
cotangent bundle of a weakly exponential Lie group. An appropriate integral
formula for the star-product is introduced together with a suitable space of
functions on which the star-product is well defined. This space of functions
becomes a Fr\'echet algebra as well as a pre-$C^*$-algebra. Basic properties of
the star-product are proved and the extension of the star-product to a Hilbert
algebra and an algebra of distributions is given. A \nbr{C^*}algebra of
observables and a space of states are constructed. Moreover, an operator
representation in position space is presented. Finally, examples of weakly
exponential Lie groups are given.
\\[\baselineskip]
\end{abstract}

\section{Introduction}
\label{sec:1}
In the paper we develop a complete theory of quantization of a classical
Hamiltonian system whose configuration space is in the form of a Lie group $G$.
The phase space $M$ of such system has the form of the cotangent bundle $\T^*G$
of $G$. A typical example of such system is a rigid body. The possible
configurations of a rigid body can be described by translations and rotations.
Thus, as the configuration space of this system we can take the group of
translations $(\mathbb{R}^3,+)$, if we are only interested in translational
degrees of freedom, or the group of rotations $SO(3)$, if we want to consider
only rotational degrees of freedom, or the semi-direct product of these two
groups if we want to take into account both translational and rotational degrees
of freedom.

Our approach to quantization is based on a deformation quantization theory
\cite{Bayen:1975-1977,Bayen:1978a,Bayen:1978b,Zachos:2005,Gosson:2006,%
Blaszak:2012}. In this theory we want to describe quantization as a continuous
deformation of a classical Hamiltonian system, with respect to some deformation
parameter. The deformation parameter is usually interpreted as the Planck's
constant $\hbar$, although in reality we should take as the deformation
parameter some dimensionless combination of the Planck's constant and a
parameter characteristic of the physical system under study. In the limit
$\hbar \to 0$ we should recover the classical system. In the rest of the paper
$\hbar$ will be a fixed non-zero real number. In fact it can be negative as the
developed theory still makes sense in such case.

The idea of deformation quantization relies on the fact that all information
about a classical Hamiltonian system is encoded in its Poisson algebra
$C^\infty(M)$. The algebra $C^\infty(M)$ is the vector space of smooth
\nbr{\mathbb{C}}valued functions defined on the phase space $M$, endowed with
a point-wise product of functions, a Poisson bracket $\{\sdot,\sdot\}$ and an
involution being the complex conjugation. By deforming the algebra $C^\infty(M)$
(or its subalgebra) to a certain noncommutative Poisson algebra we get a quantum
Hamiltonian system. The deformation is such that the point-wise product is
deformed to some noncommutative product $\star$ (the so called star-product)
and the Poisson bracket is deformed to the following Lie bracket
\begin{equation}
\lshad f,g \rshad = \frac{1}{i\hbar}(f \star g - g \star f).
\end{equation}
Moreover, there should hold
\begin{equation}
f \star g \to f \cdot g, \quad \lshad f,g \rshad \to \{f,g\} \quad
\text{as $\hbar \to 0$}.
\end{equation}

In the literature can be found many works on the topic of deformation
quantization on the cotangent bundle of a Riemannian manifold
\cite{Underhill:1978,Liu:1992,Pflaum:1998,Pflaum:1999,%
Bordemann.Neumaier.Waldmann:1998,Bordemann.Neumaier.Waldmann:1999,Karasev:2005},
however most of these works do not provide a complete theory of quantization. In
\cite{Pflaum:1998,Pflaum:1999} authors constructed a quantization scheme with
the use of a symbol calculus for pseudodifferential operators on a Riemannian
manifold. Then with its help they introduced a star-product on a suitable space
of phase space functions. In \cite{Bordemann.Neumaier.Waldmann:1998,%
Bordemann.Neumaier.Waldmann:1999} authors introduced a formal star-product on
the cotangent bundle of a Riemannian manifold and constructed a representation
of the received star product algebra by means of differential operators.
However, in all these papers the received algebra of functions consisted of
formal power series in $\hbar$, and therefore did not describe a physical
system, or was not big enough to include a pre-$C^*$-algebra, which then could
be used to define states of the quantum system in a natural way.

Note, that in the approach of strict deformation quantization
\cite{Rieffel:1989a,Rieffel:1989b,Rieffel:1994,Natsume:2001,Natsume:2003} these
problems do not appear. This approach is based on a construction of a continuous
field of \nbr{C^*}algebras by means of an operator representation of phase space
functions by compact operators on a suitable Hilbert space. However, in general
it is not possible to introduce a non-formal star product on the phase space
and it is not easy to extend the \nbr{C^*}algebra of abservables to a bigger
algebra containing all functions interesting from a physical point of view.

For all these reasons we develop a theory of deformation quantization in which
we introduce a non-formal star-product on a suitable pre-$C^*$-algebra of phase
space functions and extend this algebra to include many functions appearing in
physics. The developed deformation quantization theory is, in fact, an example
of strict deformation quantization. Moreover, we want to introduce
quantization without referring to an operator representation and later construct
this representation as a byproduct of the developed theory. Such approach seems
to be more natural because we then receive quantum theory as a direct
deformation of classical mechanics. In the current work we will be dealing only
with the case of the phase space in the form of the cotangent bundle of a Lie
group, which is much simpler to work with than the cotangent bundle of a general
Riemannian manifold.

There are not many works devoted to deformation quantization
on the cotangent bundle of a Lie group. Worth noting are papers
\cite{Cahen:1982,Gutt:1983} where authors introduce a star-product on the
cotangent bundle of a Lie group in terms of a formal power series in $\hbar$.
Also in \cite{Tounsi:2003,Tounsi:2010} are studied homogeneous star-products on
cotangent bundles of Lie groups, which are characterized in terms of integral
formulas. More often star-products are investigated on duals of Lie algebras
(see e.g.\ \cite{Arnal:1999,Dito:1999,Kathotia:2000,BenAmar:2003,Gutt:2011}).
Worth noting is also the case of a deformation quantization on the cotangent
bundle of a homogeneous space (a manifold $\mathcal{Q}$ on which a Lie group
$G$ acts transitively as a group of symmetries). Such case was studied in
\cite{Landsman:1993a} from the perspective of strict deformation quantization.

Some ideas of our theory are based on the work \cite{Karasev:2005} where authors
consider non-formal deformation quantization on the cotangent bundle of a
Riemannian manifold $\mathcal{Q}$. In this paper the exponential map $\exp_q$ on
$\mathcal{Q}$ is required to be a diffeomorphism on the whole tangent space
$\T_q\mathcal{Q}$ onto the whole manifold $\mathcal{Q}$. This condition greatly
restricts admissible manifolds $\mathcal{Q}$. In our approach we replace
$\mathcal{Q}$ by a Lie group $G$ with the property that its exponential map
$\exp$ is a diffeomorphism onto almost the whole group $G$. Many Lie groups
interesting from a point of view of physics are of such form.

In practice we do not need the whole Poisson algebra $C^\infty(M)$ to describe
a Hamiltonian system but only some subalgebra $\mathcal{F}_0(M)$. For the
purpose of quantization we take $\mathcal{F}_0(M)$ as the space of functions on
$M$ which momentum Fourier transforms are smooth and compactly supported. The
algebra $\mathcal{F}_0(M)$ fully characterizes the classical Hamiltonian system.
Indeed, all information about the Poisson manifold $M$ is encoded in
$\mathcal{F}_0(M)$. Moreover, states can be defined as appropriate linear
functionals on $\mathcal{F}_0(M)$. In fact $\mathcal{F}_0(M)$ is a dense
subalgebra of the \nbr{C^*}algebra $\mathcal{A}_0(M) := C_0(M)$ of continuous
functions on $M$ vanishing at infinity, with the standard supremum norm.
In other words $\mathcal{F}_0(M)$ is a pre-$C^*$-algebra. States are then
defined as continuous linear functionals $\Lambda$ on $\mathcal{A}_0(M)$
satisfying
\begin{enumerate}[(i)]
\item $\norm{\Lambda} = 1$ (normalization),
\item $\Lambda(\bar{f} \cdot f) \geq 0$ for every $f \in \mathcal{A}_0(M)$
(positive-definiteness).
\end{enumerate}
The Riesz representation theorem provides identification of states $\Lambda$
with probabilistic Borel measures $\mu$ through the following formula
\begin{equation}
\Lambda(f) = \int_M f \ud{\mu}.
\end{equation}
The deformation of the Poisson algebra $\mathcal{F}_0(M)$ will be denoted by
$\mathcal{F}_\hbar(M)$. Quite often we will skip the dependence on $\hbar$ and
just write $\mathcal{F}(M)$.

The paper is organized as follows. In Section~\ref{sec:2} we introduce notation
and definitions used throughout the paper. In particular, in
Section~\ref{subsec:2.6} we define the deformed Poisson algebra
$\mathcal{F}(M)$.

In Section~\ref{subsec:3.1} we introduce on $\mathcal{F}(M)$ a noncommutative
\nbr{\star}product in terms of an integral formula. Basic properties of the
\nbr{\star}product are proved and its expansion in $\hbar$ to the third order is
calculated. The algebra $\mathcal{F}(M)$ is a basis of the whole theory. With
its help we can define states of the quantum system. First, in
Section~\ref{subsec:3.2} we extend $\mathcal{F}(M)$ to a Hilbert algebra
$\mathcal{L}(M)$ and then in Section~\ref{subsec:3.3} we extend $\mathcal{L}(M)$
to a \nbr{C^*}algebra of observables $\mathcal{A}(M)$. Quantum states can be
defined in a standard way as continuous positively defined linear functionals on
$\mathcal{A}(M)$ normalized to unity. We also provide characterization of states
as elements of $\mathcal{L}(M)$ satisfying certain properties.

On the other hand the algebra $\mathcal{F}(M)$ can be extended to an algebra of
distributions $\mathcal{F}_\star(M)$ by treating $\mathcal{F}(M)$ as the space
of test functions (see Section~\ref{subsec:3.4}). The algebra
$\mathcal{F}_\star(M)$ contains all functions interesting from a physical point
of view. In particular we show that all smooth functions polynomial in fiber
variables are elements of $\mathcal{F}_\star(M)$.

For completeness of the quantization procedure we present in
Section~\ref{subsec:3.5} a short description of the time evolution of
states and observables in the language of deformation quantization.

In Section~\ref{sec:4} we construct an operator representation of the developed
theory. We represent considered in previous sections algebras of functions on
the phase space $M$ by algebras of operators on a Hilbert space $L^2(G,\dd{m})$
of \nbr{\mathbb{C}}valued functions on the Lie group $G$ square integrable with
respect to a Haar measure $\dd{m}$. We give explicit formulas of the constructed
representation. The introduced operator representation allows to express a
quantum system in the language of the standard Hilbert space approach to quantum
mechanics.

The action of the group $G$ on the phase space $M = \T^*G$ gives rise to a
reduced system whose phase space is the dual $\mathfrak{g}^*$ of the Lie algebra
$\mathfrak{g}$ of the Lie group $G$, endowed with the standard Poisson
structure. In Section~\ref{sec:5} we describe this reduction operation and show
that the introduced quantization procedure respects this reduction. In
particular, we give a formula for the \nbr{\star}product on $\mathfrak{g}^*$.

The introduced theory of quantization works only on a certain type of Lie groups
which we call weakly exponential Lie groups. In Section~\ref{subsec:2.1} we
define the notion of a weakly exponential Lie group and in Section~\ref{sec:6}
we give examples of such groups.

We end the paper with Section~\ref{sec:7} where we give conclusions and final
remarks.

\section{Preliminaries}
\label{sec:2}
\subsection{Lie groups basics}
\label{subsec:2.1}
Let $G$ be an \nbr{n}dimensional Lie group. Its Lie algebra will be denoted by
$\mathfrak{g}$. Moreover, we will denote by $\dd{m}$ a left invariant Haar
measure on $G$.

We will denote by $L_q$ and $R_q$ left and right translations in $G$, i.e.
\begin{equation}
\begin{gathered}
L_q \colon G \to G, \quad h \mapsto L_q(h) := qh, \\
R_q \colon G \to G, \quad h \mapsto R_q(h) := hq.
\end{gathered}
\end{equation}
The derivative of $L_q$ in the unit element $e$ is a linear isomorphism
$\T_e L_q \colon \mathfrak{g} \to \T_q G$ which provides a natural
identification of a tangent space $\T_q G$ with the Lie algebra $\mathfrak{g}$.
The transposition of this map $\T_e^* L_q \colon \T_q^* G \to \mathfrak{g}^*$
provides a natural identification of a cotangent space $\T_q^* G$ with the dual
of the Lie algebra $\mathfrak{g}$.

The mappings $L \colon G \times \mathfrak{g} \to \T G$ and
$\theta \colon G \times \mathfrak{g}^* \to \T^*G$, given by
\begin{equation}
\begin{gathered}
L_X(q) = \T_e L_q X, \quad q \in G, X \in \mathfrak{g}, \\
\theta_p(q) = (\T_e^* L_q)^{-1}p, \quad q \in G, p \in \mathfrak{g}^*,
\end{gathered}
\end{equation}
provide natural isomorphisms between tangent and cotangent bundles of $G$, and
trivial bundles $G \times \mathfrak{g}$ and $G \times \mathfrak{g}^*$,
respectively. The diffeomorphism $\theta$ can be used to transfer a natural
structure of a symplectic manifold on $\T^*G$ to $G \times \mathfrak{g}^*$. In
what follows we will denote the symplectic manifold $G \times \mathfrak{g}^*$
by $M$.

Recall that the derivative of the exponential map is equal
\cite{Duistermaat:2000}
\begin{equation}
\T_X \exp = \T_e L_{\exp(X)} \circ \phi(\ad_X)
= \T_e R_{\exp(X)} \circ \phi(-\ad_X), \quad
\phi(x) = \frac{1 - e^{-x}}{x} = \sum_{n=0}^\infty \frac{(-1)^n}{(n+1)!} x^n.
\end{equation}
Since $\phi(x)$ is analytic on $\mathbb{C}$ the linear operator $\phi(\ad_X)$
is well defined for every $X \in \mathfrak{g}$ as a convergent series of linear
operators. We calculate that
\begin{subequations}
\label{eq:1}
\begin{align}
\T_X (L_q \circ \exp) & = \T_{\exp(X)} L_q \circ \T_e L_{\exp(X)} \circ
    \phi(\ad_X)
= \T_e (L_q \circ L_{\exp(X)}) \circ \phi(\ad_X) \nonumber \\
& = \T_e L_{q\exp(X)} \circ \phi(\ad_X), \label{eq:1a} \\
\T_X (R_q \circ \exp) & = \T_{\exp(X)} R_q \circ \T_e R_{\exp(X)} \circ
    \phi(-\ad_X)
= \T_e (R_q \circ R_{\exp(X)}) \circ \phi(-\ad_X) \nonumber \\
& = \T_e R_{\exp(X)q} \circ \phi(-\ad_X). \label{eq:1b}
\end{align}
\end{subequations}

Let us define the reflection about the unit element $e$ as the following map
\begin{equation}
s \colon G \to G, \quad g \mapsto s(g) := g^{-1}.
\end{equation}
Moreover, define the reflection about point $q \in G$ as the map
\begin{equation}
s_q \colon G \to G, \quad s_q := L_q \circ s \circ L_{q^{-1}}
= R_q \circ s \circ R_{q^{-1}}, \quad \text{i.e.}
\quad s_q(g) = q(q^{-1}g)^{-1} = qg^{-1}q.
\label{eq:46}
\end{equation}
Note, that $s_q^2 = \id$. Denote the derivative of the reflection $s$ in the
unit element $e$ by $s'$, i.e.
\begin{equation}
s' := \T_e s, \quad s' \colon \mathfrak{g} \to \mathfrak{g}, \quad s'(X) = -X.
\end{equation}

Throughout the paper we will consider Lie groups satisfying
the following condition: there exists an open neighborhood $\mathcal{O}$ of 0 in
$\mathfrak{g}$, such that:
\begin{enumerate}[(i)]
\item it is star-shaped and symmetric, i.e.\ if $X \in \mathcal{O}$ then
$tX \in \mathcal{O}$ for $-1 \leq t \leq 1$,
\item the exponential map $\exp \colon \mathfrak{g} \to G$ restricted to
$\mathcal{O}$ is a diffeomorphism onto $\mathcal{U} = \exp(\mathcal{O})$,
\item $G \setminus \mathcal{U}$ is of measure zero.
\end{enumerate}
Lie groups satisfying this property will be called weakly exponential. Note,
however, that some authors use the term weakly exponential Lie group in a
different sense, namely as such Lie group for which $\exp(\mathfrak{g})$ is
dense in $G$. Throughout the paper $\mathcal{O}$ and $\mathcal{U}$ will denote
sets, corresponding to a given Lie group $G$, as in the above definition.

Recall that a subset $A$ of a smooth \nbr{n}dimensional manifold $M$ has measure
zero if for every smooth chart $(U,\varphi)$ for $M$, the set
$\varphi(A \cap U)$ has Lebesgue-measure zero in $\mathbb{R}^n$ \cite{Lee:2003}.

Note, that any set of measure zero has dense complement, because if
$M \setminus A$ is not dense, then $A$ contains a nonempty open set, which
would imply that $\psi(A \cap V)$ contains a nonempty open set of
Lebesgue-measure zero for some smooth chart $(\psi,V)$, which is impossible
because the only open set in $\mathbb{R}^n$ of Lebesgue-measure zero is an
empty set. Hence if $G$ is weakly exponential then $\mathcal{U}$ is dense in
$G$.

Observe, that since $\mathcal{O}$ is reflection invariant
($\mathcal{O} = -\mathcal{O}$), then $s(\mathcal{U}) = \mathcal{U}$ and
$s_q(L_q(\mathcal{U})) = L_q(\mathcal{U})$. Note also that weakly exponential
Lie groups are necessarily connected \cite{Dokovic:1997}.

We will also assume that considered weakly exponential Lie groups are such that
$\exp^{-1} \colon \mathcal{U} \to \mathcal{O}$ maps precompact sets onto
precompact sets, i.e.\ if $S \subset \mathcal{U}$ and its closure $\overline{S}$
is compact in $G$, then $\overline{\exp^{-1}(S)}$ is compact in $\mathcal{O}$.

For weakly exponential Lie groups $G$ satisfying this property we immediately
get that if $G$ is compact, then the closure $\overline{\mathcal{O}}$ of
$\mathcal{O}$ is also compact and $\exp(\overline{\mathcal{O}}) = G$.

\subsection{Poisson structure on $M$}
\label{subsec:2.2}
We will give an explicit formula for the canonical Poisson bracket on $M$.
Instead of using Darboux coordinate frame on $M$, it is convenient to use
different frame fields. Let $X_1,\dotsc,X_n$ be a basis in $\mathfrak{g}$ and
$X^1,\dotsc,X^n$ a dual basis in $\mathfrak{g}^*$. Each $X_j$ defines a left
invariant vector field $L_{X_j}$ on $G$. The vector fields $L_{X_j}$ can be
lifted to vector fields $Y_j$ on $M = G \times \mathfrak{g}^*$ by means of the
identification $\T_{(q,p)}(G \times \mathfrak{g}^*) = \T_q G \oplus
\mathfrak{g}^*$, i.e.
\begin{equation}
Y_j(q,p) = L_{X_j}(q) \oplus 0.
\end{equation}
Similarly covectors $X^j$ can be identified with constant vector fields $Z^j$ on
$M = G \times \mathfrak{g}^*$ according to the prescription
\begin{equation}
Z^j(q,p) = 0 \oplus X^j.
\end{equation}
The $2n$ vector fields $Y_1,\dotsc,Y_n,Z^1,\dotsc,Z^n$ on $M$ form at each point
$x$ a basis of tangent vectors. They have the following properties:
\begin{equation}
\begin{gathered}
\pi_* Z^i = 0, \quad \pi_* Y_i = L_{X_i}, \\
[Z^i,Z^j] = [Z^i,Y_j] = 0, \quad [Y_i,Y_j] = C^k_{ij} Y_k,
\end{gathered}
\end{equation}
where $\pi \colon G \times \mathfrak{g}^* \to G$ is a canonical projection and
$C^k_{ij}$ are the structure constants of $\mathfrak{g}$ in the basis $\{X_i\}$
given by
\begin{equation}
[X_i,X_j] = C^k_{ij} X_k.
\end{equation}
Let $p_j \colon M \to \mathbb{R}$ be fiber variables defined by
$p_j(q,p) = \braket{p,X_j}$. Then $\dd p_i(Z^j) = \delta_i^j$ and
$\dd p_i(Y_j) = 0$. The Poisson bracket of two functions $f,g \in C^\infty(M)$
reads:
\begin{equation}
\{f,g\} = Z^i f Y_i g - Y_i f Z^i g + p_k C^k_{ij} Z^i f Z^j g.
\label{eq:54}
\end{equation}
In particular
\begin{equation}
\{p_i,p_j\} = C^k_{ij} p_k.
\label{eq:48}
\end{equation}
The canonical symplectic form $\omega$ on $M$ is then given by the formula
\begin{equation}
\omega = \tilde{\alpha}^i \wedge \dd{p_i}
    + \frac{1}{2}p_k C^k_{ij} \tilde{\alpha}^i \wedge \tilde{\alpha}^j,
\label{eq:56}
\end{equation}
where $\tilde{\alpha}^i = \pi^*(\alpha^i)$ is a pull-back of a left invariant
1-form $\alpha^i$ on $G$ corresponding to $X_i$, i.e.\ $\alpha^i(L_{X_j}) =
\delta^i_j$. Then $\tilde{\alpha}^i(Z^j) = 0$ and $\tilde{\alpha}^i(Y_j) =
\delta^i_j$.

\subsection{Measures and integration}
\label{subsec:2.3}
On the vector space $\mathfrak{g}$ we can introduce a measure by means of the
Haar measure $\dd{m}$ on $G$. The Haar measure $\dd{m}$ originates from a left
invariant volume form $\omega_L$ on $G$, which on the other hand is completely
determined by its value in the unit element $\omega_L(e)$. The \nbr{n}form
$\omega_L(e)$ on $\mathfrak{g}$ induces, subsequently, a measure $\dd{X}$ on
$\mathfrak{g}$.

The measure $\dd{X}$ on $\mathfrak{g}$, in turn, determines a measure $\dd{p}$
on $\mathfrak{g}^*$ by means of the Fourier transform. To be more precise, if we
define the Fourier transform of functions on $\mathfrak{g}$ by the formula
\begin{equation}
\tilde{f}(p) = \int_{\mathfrak{g}} f(X) e^{-i\braket{p,X}} \ud{X},
\quad p \in \mathfrak{g}^*,
\end{equation}
then we choose the measure $\dd{p}$ on $\mathfrak{g}^*$ so that the inverse
Fourier transform will be given by the equation
\begin{equation}
f(X) = \frac{1}{(2\pi)^n} \int_{\mathfrak{g}^*} \tilde{f}(p)
    e^{i\braket{p,X}} \ud{p}.
\end{equation}

On the symplectic manifold $M = G \times \mathfrak{g}^*$ we have a distinguished
measure $\dd{x}$ induced by the Liouville volume form $\Omega$ on $M$.
The measure $\dd{x}$ is the product measure and can be written in the form
\begin{equation}
\dd{x} = \ud{m(q)} \times \dd{p}.
\end{equation}
Indeed, using \eqref{eq:56} we calculate that
\begin{equation}
\Omega = \frac{1}{n!} \underbrace{\omega \wedge \dotsm \wedge \omega}_n
= \tilde{\alpha}^1 \wedge \dd{p_1} \wedge \dotsm \wedge
    \tilde{\alpha}^n \wedge \dd{p_n}
= (-1)^{n(n+1)/2} \tilde{\alpha}^1 \wedge \dotsm \wedge \tilde{\alpha}^n \wedge
    \dd{p_1} \wedge \dotsm \wedge \dd{p_n}.
\end{equation}
We have that
\begin{equation}
\begin{aligned}
\tilde{\alpha}^1 \wedge \dotsm \wedge \tilde{\alpha}^n & =
    \pi_1^*(\alpha^1) \wedge \dotsm \wedge \pi_1^*(\alpha^n)
= \pi_1^*(\alpha^1 \wedge \dotsm \wedge \alpha^n), \\
\dd{p_1} \wedge \dotsm \wedge \dd{p_n} & =
    \pi_2^*(X_1) \wedge \dotsm \wedge \pi_2^*(X_n)
= \pi_2^*(X_1 \wedge \dotsm \wedge X_n),
\end{aligned}
\end{equation}
where $X_i$ are treated as constant 1-forms on $\mathfrak{g}^*$ and
$\pi_1 \colon G \times \mathfrak{g}^* \to G$, $\pi_2 \colon G \times
\mathfrak{g}^* \to \mathfrak{g}^*$ are projections onto the first and second
component of the product manifold $G \times \mathfrak{g}^*$. Thus $\Omega$ is
a wedge product of a pull-back of a left invariant volume form
$\alpha^1 \wedge \dotsm \wedge \alpha^n$ on $G$ and a pull-back of a constant
volume form $X_1 \wedge \dotsm \wedge X_n$ on $\mathfrak{g}^*$. Therefore, the
measure corresponding to $\Omega$ will be a product of a Haar measure on $G$ and
a canonical measure on $\mathfrak{g}^*$. Note, that the Liouville measure
$\dd{x}$ is independent on the normalization of the Haar measure $\dd{m}$ and
the measure $\dd{p}$ corresponding to it, since $\dd{p}$ scales inversely
proportionally to $\dd{m}$.

On the manifold $G \times \mathfrak{g}$ we will introduce the following measure
\begin{equation}
\dd{n(q,X)} = \ud{m(q)} \times \dd{X}.
\end{equation}
We will often use a normalized Liouville measure
$\dd{l} = \frac{\dd{x}}{\abs{2\pi\hbar}^n}$.

For function $f$ defined on $M = G \times \mathfrak{g}^*$ we define the Fourier
transform of $f$ in the momentum variable by
\begin{equation}
\tilde{f}(q,X) = \frac{1}{\abs{2\pi\hbar}^n} \int_{\mathfrak{g}^*}
    f(q,p) e^{\frac{i}{\hbar}\braket{p,X}} \ud{p}
\end{equation}
and the inverse Fourier transform by
\begin{equation}
f(q,p) = \int_{\mathfrak{g}} \tilde{f}(q,X)
    e^{-\frac{i}{\hbar}\braket{p,X}} \ud{X}.
\end{equation}

Quite often we will consider integrals of functions $f \in L^2(M,\dd{l})$ whose
momentum Fourier transforms $\tilde{f}$ are compactly supported, bounded and
smooth on $G \times \mathcal{O}$. Such functions may not be Lebesgue integrable,
however, we can define their integrals as improper integrals. Let
$\{K_j \mid j = 1,2,\dotsc\}$ be a sequence of compact subsets of
$\mathfrak{g}^*$ such that $K_1 \subset K_2 \subset \dotsb \subset
\mathfrak{g}^*$ and $\bigcup_{j=1}^\infty K_j = \mathfrak{g}^*$. Then the
improper integral of the function $f$ is defined as the following limit of
integrals:
\begin{equation}
\int_M f(x) \ud{l(x)} = \lim_{j \to \infty} \int_{G \times K_j} f(x) \ud{l(x)}.
\end{equation}
Note, that because of the equality
\begin{equation}
\int_{K_j} f(q,p) \ud{p} = \int_{\mathfrak{g}^*} \chi_{K_j}(p) f(q,p) \ud{p}
= \abs{2\pi\hbar}^n (\tilde{\chi}_{K_j} * \tilde{f})(q,0),
\end{equation}
where $\chi_{K_j}$ is the characteristic function of the set $K_j$ and $*$
denotes the usual convolution of functions, we get that
\begin{equation}
\int_M f(x) \ud{l(x)} = \lim_{j \to \infty} \int_G
    (\tilde{\chi}_{K_j} * \tilde{f})(q,0)
= \int_G (\delta * \tilde{f})(q,0) = \int_G \tilde{f}(q,0).
\end{equation}
This shows that the improper integral of functions in the above form always
exists, is finite and does not depend on the sequence of compact sets $\{K_j\}$.

\subsection{Baker-Campbell-Hausdorff product}
\label{subsec:2.4}
On the Lie algebra $\mathfrak{g}$ we define operation $\diamond$ by the formula
\begin{equation}
X \diamond Y = \exp^{-1}\bigl(\exp(X)\exp(Y)\bigr).
\end{equation}
It is well defined on $\mathcal{V} = \{(X,Y) \in \mathcal{O} \times \mathcal{O}
\mid \exp(X)\exp(Y) \in \mathcal{U}\}$. Clearly $X \diamond Y \in \mathcal{O}$
for $(X,Y) \in \mathcal{V}$ and $(\mathcal{O} \times \mathcal{O}) \setminus
\mathcal{V}$ is of measure zero. The operation $\diamond$ has the following
easy to verify properties:
\begin{enumerate}[(i)]
\item $X \diamond (Y \diamond Z) = (X \diamond Y) \diamond Z$ whenever
$(X,Y \diamond Z),(X \diamond Y,Z) \in \mathcal{V}$ (associativity),
\item $0 \diamond X = X \diamond 0 = X$ (0 is a neutral element),
\item $X \diamond (-X) = (-X) \diamond X = 0$ ($-X$ is an inverse element to
$X$),
\item $(-X) \diamond (-Y) = -(Y \diamond X)$,
\item $\exp(X \diamond Y) = \exp(X) \exp(Y)$.
\end{enumerate}

For a fixed element $Y \in \mathcal{O}$ let $\mathcal{O}_Y = \{X \in \mathcal{O}
\mid (X,Y) \in \mathcal{V}\}$. The maps $\mathcal{R}_Y \colon \mathcal{O}_Y \to
\mathcal{O}_{-Y}$ and $\mathcal{L}_Y \colon \mathcal{O}_{-Y} \to \mathcal{O}_Y$
given by
\begin{equation}
\mathcal{R}_Y(X) = X \diamond Y, \quad
\mathcal{L}_Y(X) = Y \diamond X
\label{eq:57}
\end{equation}
are one-to-one with inverses equal $\mathcal{R}_{-Y} \colon \mathcal{O}_{-Y} \to
\mathcal{O}_Y$ and $\mathcal{L}_{-Y} \colon \mathcal{O}_Y \to \mathcal{O}_{-Y}$
respectively. Note, that $\mathcal{O} \setminus \mathcal{O}_Y$ is of measure
zero.

In general $\diamond$ is not commutative and bilinear. By virtue of the
Baker-Campbell-Hausdorff formula we get the following expansion of
$X \diamond Y$ around $(0,0)$ up to third order
\begin{equation}
X \diamond Y = X + Y + \frac{1}{2}[X,Y] + \frac{1}{12}([X,[X,Y]] + [Y,[Y,X]])
    - \frac{1}{24}[Y,[X,[X,Y]]] + \dotsb.
\label{eq:17}
\end{equation}
We will call the operation $\diamond$ a Baker-Campbell-Hausdorff product.

\subsection{Notation}
\label{subsec:2.5}
For any $q \in G$, we use the notations:
\begin{itemize}
\item $V_q = (L_q \circ \exp|_{\mathcal{O}})^{-1}$, so that
$V_q(a) = \exp^{-1}(L_q^{-1}(a)) = \exp^{-1}(q^{-1}a)$ for
$a \in L_q(\mathcal{U})$,
\item $\frac{\dd m(s_q)}{\dd m}$ is the Jacobian obtained by transforming the
measure $\dd{m}$ under the diffeomorphism $s_q$,
\item $\frac{\dd m(\exp)}{\dd X}$ is the Jacobian obtained by transforming the
measure $\dd{m}$ under the diffeomorphism $\exp|_{\mathcal{O}}$,
\item $\chi_{\mathcal{O}}$ is the characteristic function of the set
$\mathcal{O}$, i.e.
\begin{equation*}
\chi_{\mathcal{O}}(X) = \begin{cases}
    1 & \text{for $X \in \mathcal{O}$} \\
    0 & \text{for $X \notin \mathcal{O}$}
\end{cases}
\end{equation*}
\item $j_q(a) = \abs{\det \psi(\ad_{V_q(a)})} \chi_{\mathcal{O}}(2V_q(a))$ for
$a \in L_q(\mathcal{U})$, where $\psi(x) = \frac{x}{2}\coth\frac{x}{2}
= \sum_{n=0}^\infty \frac{B_{2n} x^{2n}}{(2n)!}$ and $B_n$ is the \nbr{n}th
Bernoulli number,
\item $J_q(a) = j_q(a) \frac{\dd m(s_q)}{\dd m} \bigg|_a$ for
$a \in L_q(\mathcal{U})$,
\item $F(X) = \sqrt{\abs{\det \phi(\ad_X) \det \Ad_{\exp(\frac{1}{2}X)}}}
= \sqrt{\abs{\det \lambda(\ad_X)}}$ for $X \in \mathfrak{g}$, where
$\lambda(x) = \phi(x) e^{\frac{1}{2}x} = \frac{2}{x} \sinh\frac{x}{2}
= \sum_{n=0}^\infty \frac{x^{2n}}{4^n (2n+1)!}$,
\item $L(X,Y) = \begin{cases} \frac{F(X)F(Y)}{F(X \diamond Y)} & \text{for }
(X,Y) \in \mathcal{V} \\ 0 & \text{for } (X,Y) \notin \mathcal{V} \end{cases}$.
\end{itemize}
Note, that
\begin{equation}
V_q(a) = -V_q(s_q(a)) = -V_a(q)
\end{equation}
since
\begin{align}
V_q \circ s_q & = (L_q \circ \exp)^{-1} \circ s_q
= (s_q \circ L_q \circ \exp)^{-1}
= (L_q \circ s \circ \exp)^{-1} \nonumber \\
& = (L_q \circ \exp \circ s')^{-1}
= s' \circ (L_q \circ \exp)^{-1}
= s' \circ V_q
\end{align}
and
\begin{equation}
V_q(s_q(a)) = \exp^{-1}(q^{-1} s_q(a)) = \exp^{-1}(a^{-1} q) = V_a(q).
\end{equation}
Hence, $j_q$ is invariant under the reflection $s_q$.

We have that
\begin{align}
\frac{\dd m(\exp)}{\dd X} \bigg|_X & =
    \frac{\dd m(L_q \circ \exp)}{\dd X} \bigg|_X
= \abs{\det \phi(\ad_X)}, \label{eq:58} \\
\frac{\dd m(s)}{\dd m} \bigg|_q & = \det \Ad_q, \\
\frac{\dd m(s_q)}{\dd m} \bigg|_a & = \frac{\dd m(s)}{\dd m} \bigg|_{q^{-1}a}
= \det \Ad_{q^{-1}a}.
\end{align}
Indeed, if we choose a basis $X^1,\dotsc,X^n$ in $\mathfrak{g}^*$, then the left
invariant volume form $\omega_L$ on $G$ corresponding to the Haar measure
$\dd{m}$ can be expressed by the formula:
$\omega_L = c \alpha^1 \wedge \dotsm \wedge \alpha^n$,
where $\alpha^i(a) = (\T_e^* L_a)^{-1} X^i$ is a left invariant 1-form induced
by $X^i$ and $c \in \mathbb{R}$ is some constant. The volume form on
$\mathfrak{g}$ corresponding to the measure $\dd{X}$ takes then the form
$c X^1 \wedge \dotsm \wedge X^n$. Let $\Phi = L_q \circ \exp$. The pull-back of
the 1-form $\alpha^i$ by the map $\Phi$ is equal
\begin{align}
\Phi^* \alpha^i(X) & = \alpha(\Phi(X)) \circ \Phi_*(X)
= \alpha(q\exp(X)) \circ \T_X(L_q \circ \exp)
= (\T_e^* L_{q\exp(X)})^{-1} X^i \circ \T_X(L_q \circ \exp) \nonumber \\
& = \left(\T_X^*(L_q \circ \exp) \circ (\T_e^* L_{q\exp(X)})^{-1}\right)X^i
= \left((\T_e L_{q\exp(X)})^{-1} \circ \T_X(L_q \circ \exp)\right)^* X^i
    \nonumber \\
& = \phi(\ad_X)^* X^i.
\end{align}
Thus, the pull-back of the volume form $\omega_L$ is equal
\begin{align}
\Phi^* \omega_L(X) & =
    c (\Phi^* \alpha^1)(X) \wedge \dotsm \wedge (\Phi^* \alpha^n)(X)
= c \phi(\ad_X)^* X^1 \wedge \dotsm \wedge \phi(\ad_X)^* X^n \nonumber \\
& = \det \phi(\ad_X) c X^1 \wedge \dotsm \wedge X^n.
\end{align}
This proves \eqref{eq:58}.

The function $F(X)$ is symmetric: $F(X) = F(-X)$. Note also, that for unimodular
groups (so in particular for all compact groups)
\begin{equation}
\frac{\dd m(s)}{\dd m}\bigg|_q = \det \Ad_q = 1
\end{equation}
for every $q \in G$, so that $F(X) = \sqrt{\abs{\det \phi(\ad_X)}}$.

\subsection{Spaces $\mathcal{F}(M)$ and $\mathcal{L}(M)$}
\label{subsec:2.6}
Let $L^2(G \times \mathcal{O},\dd{n})$ denote the Hilbert space of
\nbr{\mathbb{C}}valued square integrable functions on $G \times \mathfrak{g}$
with support in $G \times \overline{\mathcal{O}}$ (by $\overline{\mathcal{O}}$
we denote the closure of $\mathcal{O}$). The scalar product in this space is
given by
\begin{equation}
(f,g) = \int_{G \times \mathfrak{g}} \overline{f(q,X)} g(q,X) \ud{n(q,X)}.
\end{equation}
Moreover, let $\widetilde{\mathcal{F}}(M)$ be a subspace of
$L^2(G \times \mathcal{O},\dd{n})$ consisting of functions $\tilde{f}$ such that
the functions
\begin{equation}
\mathsf{f}(a,b) = \tilde{f}\bigl(a\exp(\tfrac{1}{2}V_a(b)),V_a(b)\bigr)
    F(V_a(b))^{-1},
\label{eq:5}
\end{equation}
defined on a dense subset $\{(a,b) \in G \times G\mid a^{-1}b \in \mathcal{U}\}$
of $G \times G$, extend to smooth functions with compact support defined on the
whole space $G \times G$. Note, that since $F(X)$ is continuous on
$\mathfrak{g}$ and smooth on $\mathcal{O}$ functions in
$\widetilde{\mathcal{F}}(M)$ are bounded,
have compact support and are smooth on $G \times \mathcal{O}$. If a function
$\tilde{f} \in \widetilde{\mathcal{F}}(M)$ has a support in
$G \times \mathcal{O}$, then it will be smooth on the whole set
$G \times \mathfrak{g}$. However, if the support of $\tilde{f}$ is not in
$G \times \mathcal{O}$ but only in its closure
$G \times \overline{\mathcal{O}}$, then $\tilde{f}$ may not be smooth
(or even continuous) on $G \times \mathfrak{g}$. As we will see in
Section~\ref{sec:4} the functions \eqref{eq:5} are the integral kernels of
operators in the corresponding operator representation. So that, we define
$\widetilde{\mathcal{F}}(M)$ in such a way that the corresponding space of
operators will consist of all integral operators whose integral kernels are
smooth and compactly supported.

Denote by $\mathcal{L}(M)$ and $\mathcal{F}(M)$ the images
of $L^2(G \times \mathcal{O},\dd{n})$ and $\widetilde{\mathcal{F}}(M)$ with
respect to the inverse Fourier transform in momentum variable. The set
$\mathcal{L}(M)$ obtains the structure of a Hilbert space from
$L^2(G \times \mathcal{O},\dd{n})$, where the scalar product on $\mathcal{L}(M)$
takes the form
\begin{equation}
(f,g) = \int_M \overline{f(x)} g(x) \ud{l(x)}.
\label{eq:15}
\end{equation}
The space $\mathcal{L}(M)$ is a Hilbert subspace of the space $L^2(M,\dd{l})$ of
\nbr{\mathbb{C}}valued square integrable functions on $M$. Since functions in
$\mathcal{F}(M)$ are inverse momentum Fourier transforms of bounded
functions with compact support, then it is easy to show that $\mathcal{F}(M)$
consists of smooth functions. Note, that elements of $\mathcal{F}(M)$ may not be
Lebesgue integrable. Therefore, integrals of these functions will be considered
as improper integrals defined as in Section~\ref{subsec:2.3}.

On the space $\mathcal{F}(M)$ we define a trace functional $\tr$ by the formula
\begin{equation}
\tr(f) = \int_M f(x) \ud{l(x)}.
\label{eq:25}
\end{equation}
By $\Tr$ we will denote the usual operator trace.

Inverting formula \eqref{eq:5} we can see that elements of $\mathcal{F}(M)$ are
all functions $f$ of the form
\begin{equation}
f(q,p) = \int_{\mathcal{O}} \mathsf{f}\bigl(q\exp(-\tfrac{1}{2}X),
    q\exp(\tfrac{1}{2}X)\bigr) e^{-\frac{i}{\hbar} \braket{p,X}} F(X) \ud{X}
\label{eq:44}
\end{equation}
for functions $\mathsf{f} \in C^\infty_c(G \times G)$. Therefore, we have a
one-to-one correspondence between elements of $\mathcal{F}(M)$ and
$C^\infty_c(G \times G)$. In the case $G$ is a compact group we will topologize
$\mathcal{F}(M)$ by the following family of semi-norms
\begin{equation}
\norm{f}_{k,l} = \sup_{a,b \in G} \abs{D^{k,l}\mathsf{f}(a,b)},
\end{equation}
where $k,l \in \mathbb{N}^n$ are multi-indices and $D^{k,l}$ are differential
operators expressed by left-invariant vector fields $L_{X_1},\dotsc,L_{X_n}$
corresponding to a basis $X_1,\dotsc,X_n$ in $\mathfrak{g}$, and defined
by the formula
\begin{equation}
D^{k,l} = \bigl(L_{X_1}^{(1)}\bigr)^{k_1} \dotsm \bigl(L_{X_n}^{(1)}\bigr)^{k_n}
    \bigl(L_{X_1}^{(2)}\bigr)^{l_1} \dotsm \bigl(L_{X_n}^{(2)}\bigr)^{l_n},
\end{equation}
where $L_{X_i}^{(1)}$ and $L_{X_j}^{(2)}$ are vector fields on $G \times G$
defined by
\begin{equation}
L_{X_i}^{(1)}(a,b) = L_{X_i}(a) \oplus 0, \quad
L_{X_j}^{(2)}(a,b) = 0 \oplus L_{X_j}(b)
\label{eq:30}
\end{equation}
according to the identification $\T_{(a,b)}(G \times G) = \T_a G \oplus \T_b G$.
The introduced topology on $\mathcal{F}(M)$ is independent on the choice of a
basis $\{X_i\}$ in $\mathfrak{g}$. It makes from $\mathcal{F}(M)$ a Fr\'echet
space.

Note, that since $C^\infty_c(G \times G) \subset L^2(G \times G,\dd{m} \times
\dd{m})$ we can introduce for $\mathsf{f},\mathsf{g} \in C^\infty_c(G \times G)$
their scalar product $(\mathsf{f},\mathsf{g})$. A similar consideration takes
place for the space $\mathcal{F}(M) \subset \mathcal{L}(M)$. We will prove a
useful lemma.

\begin{lemma}
\label{lem:1}
Let $f,g \in \mathcal{F}(M)$ and $\mathsf{f},\mathsf{g} \in
C^\infty_c(G \times G)$ be the corresponding integral kernels. Then,
\begin{equation}
(f,g) = (\mathsf{f},\mathsf{g}).
\end{equation}
\end{lemma}

\begin{proof}
We calculate that
\begin{align}
(\mathsf{f},\mathsf{g}) & = \int_{G \times G} \overline{\mathsf{f}(a,b)}
    \mathsf{g}(a,b) \ud{m(a)}\ud{m(b)} \nonumber \\
& = \int_G \left( \int_G
    \overline{\tilde{f}\bigl(a\exp(\tfrac{1}{2}V_a(b)),V_a(b)\bigr)}
    \tilde{g}\bigl(a\exp(\tfrac{1}{2}V_a(b)),V_a(b)\bigr)
    F(V_a(b))^{-2} \ud{m(b)} \right) \ud{m(a)} \nonumber \\
& = \int_G \left( \int_{\mathfrak{g}}
    \overline{\tilde{f}\bigl(a\exp(\tfrac{1}{2}X),X\bigr)}
    \tilde{g}\bigl(a\exp(\tfrac{1}{2}X),X\bigr)
    \det\Ad_{\exp(-\frac{1}{2}X)} \ud{m(a)} \right) \ud{X} \nonumber \\
& = \int_{G \times \mathfrak{g}} \overline{\tilde{f}(a,X)} \tilde{g}(a,X)
    \ud{m(a)}\ud{X}
= (\tilde{f},\tilde{g}) = (f,g).
\end{align}
\end{proof}

\section{Deformation quantization of the classical system}
\label{sec:3}
\subsection{Star-product on $\mathcal{F}(M)$}
\label{subsec:3.1}
On the space $\mathcal{F}(M)$ we introduce the following star-product
\begin{align}
(f \star g)(q,p) & = \int_{\mathfrak{g} \times \mathfrak{g}}
    \tilde{f}\bigl(q\exp(-\tfrac{1}{2}(X \diamond Y))\exp(\tfrac{1}{2}X),X\bigr)
    \tilde{g}\bigl(q\exp(\tfrac{1}{2}(X \diamond Y))\exp(-\tfrac{1}{2}Y),Y\bigr)
    e^{-\frac{i}{\hbar}\braket{p,X \diamond Y}} \nonumber \\
& \quad {} \times L(X,Y) \ud{X} \ud{Y}.
\label{eq:2}
\end{align}
Such star-product was already considered in \cite{Tounsi:2010} from the
perspective of formal deformation quantization. Note, that since $\tilde{f}$ and
$\tilde{g}$ have supports in $G \times \overline{\mathcal{O}}$ the
\nbr{\star}product of $f$ and $g$ is a well defined function on $M$. In what
follows we will show that $f \star g \in \mathcal{F}(M)$ so that the space
$\mathcal{F}(M)$ together with the \nbr{\star}product becomes an algebra.

For $Y \in \mathcal{O}$ let $\mathcal{R}_Y$ and $\mathcal{L}_Y$ be maps defined
in \eqref{eq:57}. With the use of \eqref{eq:1b} we can calculate Jacobians of
these transformations:
\begin{subequations}
\begin{align}
\abs{\det \T_X \mathcal{R}_Y} & = \left(\frac{F(X)}{F(X \diamond Y)}\right)^2
    \det \Ad_{\exp(-\frac{1}{2}Y)}, \\
\abs{\det \T_X \mathcal{L}_Y} & = \left(\frac{F(X)}{F(Y \diamond X)}\right)^2
    \det \Ad_{\exp(\frac{1}{2}Y)}.
\end{align}
\end{subequations}
If in \eqref{eq:2} under the integral with respect to $X$ we perform the
following change of variables: $X \to \mathcal{R}_{-Y}(X)$, then we can write
the \nbr{\star}product \eqref{eq:2} in the form
\begin{align}
(f \star g)(q,p) & = \int_{\mathfrak{g} \times \mathfrak{g}}
    \tilde{f}\bigl(q\exp(-\tfrac{1}{2}X)
    \exp\bigl(\tfrac{1}{2}(X \diamond Y)\bigr),X \diamond Y\bigr)
    \tilde{g}\bigl(q\exp(\tfrac{1}{2}X)\exp(\tfrac{1}{2}Y),-Y\bigr)
    e^{-\frac{i}{\hbar}\braket{p,X}} \nonumber \\
& \quad {} \times L(X, Y) \det \Ad_{\exp(-\frac{1}{2}Y)} \ud{X}\ud{Y}.
\label{eq:3}
\end{align}
We can introduce a twisted convolution of functions $\tilde{f},\tilde{g} \in
\widetilde{\mathcal{F}}(M)$ by the formula
\begin{align}
(\tilde{f} \odot \tilde{g})(q,X) & = \int_{\mathfrak{g}}
    \tilde{f}\bigl(q\exp(-\tfrac{1}{2}X)
    \exp\bigl(\tfrac{1}{2}(X \diamond Y)\bigr),X \diamond Y\bigr)
    \tilde{g}\bigl(q\exp(\tfrac{1}{2}X)\exp(\tfrac{1}{2}Y),-Y\bigr) \nonumber \\
& \quad {} \times L(X, Y) \det \Ad_{\exp(-\frac{1}{2}Y)} \ud{Y}.
\label{eq:4}
\end{align}
Then from \eqref{eq:3} we immediately get that
\begin{equation}
(f \star g)^\sim = \tilde{f} \odot \tilde{g}.
\label{eq:6}
\end{equation}

The following theorem gathers basic properties of the \nbr{\star}product.

\begin{theorem}
\label{thm:1}
The \nbr{\star}product on $\mathcal{F}(M)$ is a bilinear operation with the
following properties:
\begin{enumerate}[(i)]
\item\label{item:1a} $f \star g \in \mathcal{F}(M)$ (the space $\mathcal{F}(M)$
is closed with respect to the \nbr{\star}product),
\item\label{item:1b} $(f \star g) \star h = f \star (g \star h)$
(associativity),
\item\label{item:1c} $\overline{f \star g} = \bar{g} \star \bar{f}$
(complex-conjugation is an involution),
\item\label{item:1d} $\displaystyle \int_M f \star g \ud{x} =
\int_M f(x)g(x) \ud{x}$,
\end{enumerate}
for $f,g,h \in \mathcal{F}(M)$.
\end{theorem}

\begin{proof}
(\ref{item:1a}) If $\mathsf{f},\mathsf{g} \in C^\infty_c(G \times G)$ are
functions corresponding to $\tilde{f}$ and $\tilde{g}$ according to
\eqref{eq:5}, then by virtue of the equality
\begin{equation}
V_a(b) \diamond Y = V_a(b\exp(Y)),
\end{equation}
we get the following expression for a function $\mathsf{h}$ corresponding to
$\tilde{f} \odot \tilde{g}$
\begin{align}
\mathsf{h}(a,b) & =
    (\tilde{f} \odot \tilde{g})\bigl(a\exp(\tfrac{1}{2}V_a(b)),V_a(b)\bigr)
    F(V_a(b))^{-1} = \int_{\mathfrak{g}}
    \tilde{f}\bigl(a\exp(\tfrac{1}{2}V_a(b\exp(Y))),
    V_a(b\exp(Y))\bigr) \nonumber \\
& \quad \times \tilde{g}\bigl(b\exp(\tfrac{1}{2}Y),-Y\bigr)
    F(V_a(b\exp(Y)))^{-1} F(Y)^{-1}
    \frac{\dd m(L_b \circ \exp)}{\dd Y}\bigg|_Y \ud{Y} \nonumber \\
& = \int_G \tilde{f}\bigl(a\exp(\tfrac{1}{2}V_a(c)),V_a(c)\bigr)
    F(V_a(c))^{-1} \tilde{g}\bigl(c\exp(\tfrac{1}{2}V_c(b)),V_c(b)\bigr)
    F(V_c(b))^{-1} \ud{m(c)} \nonumber \\
& = \int_G \mathsf{f}(a,c) \mathsf{g}(c,b) \ud{m(c)}.
\label{eq:23}
\end{align}
The function
\begin{equation}
\mathsf{h}(a,b) = \int_G \mathsf{f}(a,c) \mathsf{g}(c,b) \ud{m(c)}
\label{eq:59}
\end{equation}
is smooth and compactly supported. Indeed, $\supp(\mathsf{f}) \subset
\pi_1(\supp(\mathsf{f})) \times \pi_2(\supp(\mathsf{f}))$ and similarly
$\supp(\mathsf{g}) \subset \pi_1(\supp(\mathsf{g})) \times
\pi_2(\supp(\mathsf{g}))$, where $\pi_i \colon G \times G \to G$, $i = 1,2$,
are canonical projections onto the first and second component of the product
manifold $G \times G$. Thus $\supp(\mathsf{h}) \subset \pi_1(\supp(\mathsf{f}))
\times \pi_2(\supp(\mathsf{g}))$ and since $\pi_1$ and $\pi_2$ are continuous
the support of $\mathsf{h}$ is a subset of a compact set and as such is also
compact. The value of the integral in \eqref{eq:59} does not change if we
restrict the integration to a compact set $\pi_2(\supp(\mathsf{f})) \cap
\pi_1(\supp(\mathsf{g}))$. Using this fact we can interchange the integration
with a limiting operation and easily show that partial derivatives of any order
of $\mathsf{h}$ exist and are jointly continuous as functions on $G \times G$.
Therefore, $\mathsf{h} \in C^\infty_c(G \times G)$ which further implies that
$\tilde{f} \odot \tilde{g} \in \widetilde{\mathcal{F}}(M)$ and consequently
$f \star g \in \mathcal{F}(M)$.

(\ref{item:1b}) We have that
\begin{align}
& ((\tilde{f} \odot \tilde{g}) \odot \tilde{h})(q,X) =
    \int_{\mathfrak{g}} \biggl( \int_{\mathfrak{g}}
    \tilde{f}\bigl(q\exp(-\tfrac{1}{2}X)
    \exp(\tfrac{1}{2}(X \diamond Y \diamond Z)),X \diamond Y \diamond Z\bigr)
    \nonumber \\
& \qquad \times \tilde{g}\bigl(q\exp(\tfrac{1}{2}X)\exp(Y)\exp(\tfrac{1}{2}Z),
    -Z\bigr) \tilde{h}\bigl(q\exp(\tfrac{1}{2}X)\exp(\tfrac{1}{2}Y),-Y\bigr)
    \nonumber \\
& \qquad \times L(X,Y)L(X \diamond Y,Z)
    \det\Ad_{\exp(-\frac{1}{2}Y)\exp(-\frac{1}{2}Z)} \ud{Z} \biggr) \ud{Y}
    \displaybreak[0] \nonumber \\
& \quad = \int_{\mathcal{O}} \biggl( \int_{\mathcal{O}_{-Y}}
    \tilde{f}\bigl(q\exp(-\tfrac{1}{2}X)\exp(\tfrac{1}{2}(X \diamond
    \mathcal{L}_Y(Z))),X \diamond \mathcal{L}_Y(Z)\bigr) \nonumber \\
& \qquad \times \tilde{g}\bigl(q\exp(\tfrac{1}{2}X)\exp(\mathcal{L}_Y(Z))
    \exp\bigl(\tfrac{1}{2}((-\mathcal{L}_Y(Z)) \diamond Y)\bigr),
    (-\mathcal{L}_Y(Z)) \diamond Y\bigr) \nonumber \\
& \qquad \times \tilde{h}\bigl(q\exp(\tfrac{1}{2}X)\exp(\tfrac{1}{2}Y),-Y\bigr)
    L(-\mathcal{L}_Y(Z),Y)L(X,\mathcal{L}_Y(Z)) \nonumber \\
& \qquad \times \det\Ad_{\exp(-\frac{1}{2}Y)\exp(-\frac{1}{2}\mathcal{L}_Y(Z))}
    \abs{\det \T_Z \mathcal{L}_Y} \ud{Z} \biggr) \ud{Y} \displaybreak[0]
    \nonumber \\
& \quad = \int_{\mathfrak{g}} \biggl( \int_{\mathfrak{g}}
    \tilde{f}\bigl(q\exp(-\tfrac{1}{2}X)
    \exp(\tfrac{1}{2}(X \diamond Z)),X \diamond Z\bigr)
    \tilde{g}\bigl(q\exp(\tfrac{1}{2}X)\exp(Z)
    \exp\bigl(\tfrac{1}{2}((-Z) \diamond Y)\bigr),(-Z) \diamond Y\bigr)
    \nonumber \\
& \qquad \times \tilde{h}\bigl(q\exp(\tfrac{1}{2}X)\exp(\tfrac{1}{2}Y),-Y\bigr)
    L(-Z,Y)L(X,Z) \det\Ad_{\exp(-\frac{1}{2}Y)\exp(-\frac{1}{2}Z)} \ud{Y}
    \biggr) \ud{Z} \nonumber \\
& \quad = (\tilde{f} \odot (\tilde{g} \odot \tilde{h}))(q,X).
\end{align}
Applying \eqref{eq:6} yields the associativity of $\star$.

(\ref{item:1c}) This property follows directly from the definition \eqref{eq:2}
of the \nbr{\star}product.

(\ref{item:1d}) We calculate that
\begin{align}
\int_M f \star g \ud{x} & = \int_G (f \star g)^\sim(q,0) \ud{m(q)}
= \int_G (\tilde{f} \odot \tilde{g})(q,0) \ud{m(q)} \nonumber \\
& = \int_G \left(\int_{\mathfrak{g}}\tilde{f}\bigl(q\exp(\tfrac{1}{2}Y),Y\bigr)
    \tilde{g}\bigl(q\exp(\tfrac{1}{2}Y),-Y\bigr) \det\Ad_{\exp(-\frac{1}{2}Y)}
    \ud{Y} \right) \ud{m(q)} \nonumber \\
& = \int_G \left( \int_{\mathfrak{g}} \tilde{f}(q,Y) \tilde{g}(q,-Y) \ud{Y}
    \right) \ud{m(q)}
= \int_G (\tilde{f} * \tilde{g})(q,0) \ud{m(q)} = \int_M f(x)g(x) \ud{x},
\end{align}
where $*$ denotes the usual convolution.
\end{proof}

From \eqref{eq:23} we can see that the integral kernel $\mathsf{h} \in
C^\infty_c(G \times G)$ corresponding to the product $f \star g$ is expressed
by the integral kernels $\mathsf{f},\mathsf{g} \in C^\infty_c(G \times G)$ of
$f$ and $g$ by the formula
\begin{equation}
\mathsf{h}(a,b) = \int_G \mathsf{f}(a,c) \mathsf{g}(c,b) \ud{m(c)}.
\label{eq:29}
\end{equation}
From this, in the case of a compact group $G$, we get the following estimate
on the semi-norms of the product $f \star g$:
\begin{equation}
\norm{f \star g}_{k,l} \leq m(G) \norm{f}_{k,0} \norm{g}_{0,l},
\end{equation}
where $m(G)$ is the measure of the group $G$. Therefore, the \nbr{\star}product
as a map $\mathcal{F}(M) \times \mathcal{F}(M) \to \mathcal{F}(M)$ will be a
continuous operation on the Fr\'echet space $\mathcal{F}(M)$.

At this point a remark should be made concerning the dependence of the
quantization on the global topological and geometrical structure of the Lie
group $G$. It might seem that we are neglecting all global issues by working
only on an open subset $\mathcal{U}$ of the Lie group $G$, which we then map to
an open subset $\mathcal{O}$ of the Lie algebra $\mathfrak{g}$. However, the
global topological and geometrical structure of $G$ is encoded in the definition
of the algebra $\mathcal{F}(M)$. This may be readily seen from the fact that
there is a one-to-one correspondence between elements of $\mathcal{F}(M)$ and
$C^\infty_c(G \times G)$, where to the \nbr{\star}product of functions
corresponds a composition of integral kernels given by \eqref{eq:29}. Clearly
the space $C^\infty_c(G \times G)$ depends on the global structure of $G$ and
so is $\mathcal{F}(M)$.

In what follows we will investigate the dependence of $\mathcal{F}(M)$ on
$\hbar$. For this purpose we will change, for a moment, the definition of the
Fourier transform in the momentum variable in the following way
\begin{equation}
\tilde{f}(q,X) = \frac{1}{(2\pi)^n} \int_{\mathfrak{g}^*} f(q,p)
    e^{i\braket{p,X}} \ud{p},
\end{equation}
so that it will not depend on $\hbar$. Then we also have to redefine the space
$\widetilde{\mathcal{F}}(M)$ as the space of square integrable functions
$\tilde{f}$ with support in $G \times \hbar^{-1} \overline{\mathcal{O}}$ such
that the functions
\begin{equation}
\mathsf{f}(a,b) = \abs{\hbar}^{-n} \tilde{f}\bigl(a\exp(\tfrac{1}{2}V_a(b)),
    \hbar^{-1}V_a(b)\bigr) F(V_a(b))^{-1}
\end{equation}
extend to smooth functions on $G \times G$ with compact support. We will
explicitly denote the dependence of the spaces $\widetilde{\mathcal{F}}(M)$ and
$\mathcal{F}(M)$ on $\hbar$ by writing $\widetilde{\mathcal{F}}_\hbar(M)$ and
$\mathcal{F}_\hbar(M)$. Moreover, we will denote by
$\widetilde{\mathcal{F}}_0(M)$ the space $C^\infty_c(G \times \mathfrak{g})$ of
smooth functions with compact support and by $\mathcal{F}_0(M)$ its inverse
momentum Fourier transform. Then for any $f \in \mathcal{F}_0(M)$ there exists
a sufficiently small $\hbar_0 > 0$ for which $f \in \mathcal{F}_\hbar(M)$ for
every $\hbar \in \mathbb{R}$ such that $\abs{\hbar} < \hbar_0$.

Introducing a Baker-Campbell-Hausdorff product $\diamond_\hbar$ on the Lie
algebra $\mathfrak{g}$ corresponding to a Lie bracket
$[\sdot,\sdot]_\hbar = \hbar[\sdot,\sdot]$, we can write the \nbr{\star}product
in the following way
\begin{align}
(f \star_\hbar g)(q,p) & = \int_{\mathfrak{g} \times \mathfrak{g}}
    \tilde{f}\bigl(q\exp(-\tfrac{1}{2}\hbar(X \diamond_\hbar Y))
    \exp(\tfrac{1}{2}\hbar X),X\bigr)
    \tilde{g}\bigl(q\exp(\tfrac{1}{2}\hbar(X \diamond_\hbar Y))
    \exp(-\tfrac{1}{2}\hbar Y),Y\bigr)
    e^{-i\braket{p,X \diamond_\hbar Y}} \nonumber \\
& \quad {} \times L_\hbar(X,Y) \ud{X} \ud{Y},
\end{align}
where $L_\hbar(X,Y) = F_\hbar(X \diamond_\hbar Y)^{-1} F_\hbar(X) F_\hbar(Y)$
and $F_\hbar(X) = F(\hbar X)$. From this presentation of the \nbr{\star}product
it can be easily seen that for $f,g \in \mathcal{F}_0(M)$ and a sufficiently
small $\hbar_0 > 0$ the functions $f,g \in \mathcal{F}_\hbar(M)$ for every
$\hbar \in \mathbb{R}$ such that $\abs{\hbar} < \hbar_0$, and for any $x \in M$
the function $\hbar \mapsto (f \star_\hbar g)(x)$ is smooth on
$(-\hbar_0,\hbar_0)$ and $(f \star_\hbar g)(x) \to f(x)g(x)$ as $\hbar \to 0$.
From the following theorem we also get that $\lshad f,g \rshad_\hbar(x) \to
\{f,g\}(x)$ as $\hbar \to 0$, where $\{\sdot,\sdot\}$ is the canonical Poisson
bracket on $M$.

\begin{theorem}
The \nbr{\star}product enjoys the following power series expansion in $\hbar$
around $\hbar = 0$ up to third order:
\begin{align}
f \star_\hbar g & = fg + \frac{i\hbar}{2} \{f,g\}
    + \frac{1}{2!}\left(\frac{i\hbar}{2}\right)^2\bigl(B_2(f,g) + B_2(g,f)\bigr)
    + \frac{1}{3!}\left(\frac{i\hbar}{2}\right)^3\bigl(B_3(f,g) - B_3(g,f)\bigr)
    + O(\hbar^4),
\label{eq:18}
\end{align}
where
\begin{subequations}
\begin{align}
B_2(f,g) & = Z^i Z^j f Y_i Y_j g - Z^i Y_j f Z^j Y_i g
     + C^k_{ij} Z^i f Z^j Y_k g - 2 p_k C^k_{ij} Z^i Y_l f Z^l Z^j g
    - \frac{1}{6} C^k_{il} C^l_{jk} Z^i f Z^j g \nonumber \\
& \quad {} + \frac{1}{2} p_k p_l C^k_{ij} C^l_{rs} Z^i Z^r f Z^j Z^s g
    + \frac{2}{3} p_k C^k_{il} C^l_{jr} Z^i Z^j f Z^r g, \\
B_3(f,g) & = Z^i Z^j Z^k f Y_i Y_j Y_k g - 3 Z^i Z^j Y_k f Z^k Y_i Y_j g
    + 3 C^k_{ij} Z^i Z^l f Z^j Y_l Y_k g - 3 C^k_{ij} Z^i Y_l f Z^l Z^j Y_k g
    \nonumber \\
& \quad {} + C^k_{il} C^l_{jk} Z^i Y_r f Z^r Z^j g
    + 3 p_k C^k_{ij} Z^i Y_r Y_s f Z^r Z^s Z^j g
    - 3 p_k C^k_{ij} Z^i Z^r Y_s f Z^j Z^s Y_r g \nonumber \\
& \quad {} - 3 p_k C^k_{ij} C^l_{rs} Z^i Z^r Y_l f Z^j Z^s g
    - 2 p_k C^k_{il} C^l_{jr} Z^i Z^j Y_s f Z^s Z^r g
    + 2 p_k C^k_{il} C^l_{jr} Z^j Y_s f Z^s Z^r Z^i g \nonumber \\
& \quad {} - 3 p_k p_l C^k_{ij} C^l_{rs} Z^i Z^r Y_m f Z^m Z^j Z^s g
    - C^s_{lr} C^l_{si} C^r_{jk} Z^i Z^j f Z^k g
    - p_k C^k_{jl} C^l_{im} C^m_{rs} Z^i Z^r f Z^j Z^s g \nonumber \\
& \quad {} - \frac{1}{2}p_k C^k_{lm} C^l_{ij} C^m_{rs} Z^i Z^r f Z^j Z^s g
    - p_k C^k_{ij} C^m_{rl} C^l_{sm} Z^i Z^r f Z^j Z^s g
    + 2 p_k p_l C^k_{ij} C^l_{rm} C^m_{st} Z^i Z^r Z^s f Z^j Z^t g \nonumber \\
& \quad {} + \frac{1}{2} p_k p_l p_m C^k_{ij} C^l_{rs} C^m_{tu} Z^i Z^r Z^t f
    Z^j Z^s Z^u g.
\end{align}
\end{subequations}
\end{theorem}

\begin{proof}
From Taylor theorem we get for every $x \in M$ that
\begin{align}
(f \star_\hbar g)(x) & = f(x)g(x) + \frac{\dd}{\dd\hbar}(f \star_\hbar g)(x)
    \bigg|_{\hbar = 0} \hbar
    + \frac{1}{2!} \frac{\dd^2}{\dd\hbar^2}(f \star_\hbar g)(x)
    \bigg|_{\hbar = 0} \hbar^2 \nonumber \\
& \quad {} + \frac{1}{3!} \frac{\dd^3}{\dd\hbar^3}(f \star_\hbar g)(x)
    \bigg|_{\hbar = 0} \hbar^3 + O(\hbar^4).
\end{align}
Therefore, we have to calculate derivatives with respect to $\hbar$ of
$(f \star_\hbar g)(x)$ at $\hbar = 0$.

For $\varphi \in C^\infty(G)$ and a smooth function $A \colon \mathbb{R} \to
\mathfrak{g}$ such that $A(0) = 0$ we get
\begin{align}
\frac{\dd}{\dd\hbar} \varphi\bigl(q\exp(A(\hbar))\bigr) & = \T_{q\exp(A(\hbar))}
    \varphi\bigl(\T_e L_{q\exp(A(\hbar))}\phi(\ad_{A(\hbar)})A'(\hbar)\bigr)
= \T_{q\exp(A(\hbar))} \varphi\bigl(\T_e L_{q\exp(A(\hbar))} B(\hbar)\bigr)
    \nonumber \\
& = \T_{q\exp(A(\hbar))} \varphi\bigl(L_{B(\hbar)}(q\exp(A(\hbar)))\bigr)
= L_{B(\hbar)} \varphi\bigl(q\exp(A(\hbar))\bigr),
\end{align}
where $B(\hbar) = \phi(\ad_{A(\hbar)})A'(\hbar)$ and $L_{B(\hbar)}$ denotes a
left-invariant vector field corresponding to $B(\hbar) \in \mathfrak{g}$.
Moreover, we calculate that
\begin{equation}
\begin{split}
\frac{\dd^2}{\dd\hbar^2} \varphi\bigl(q\exp(A(\hbar))\bigr) & =
    L_{B(\hbar)} L_{B(\hbar)} \varphi\bigl(q\exp(A(\hbar))\bigr)
    + L_{B'(\hbar)} \varphi\bigl(q\exp(A(\hbar))\bigr), \\
\frac{\dd^3}{\dd\hbar^3} \varphi\bigl(q\exp(A(\hbar))\bigr) & =
    L_{B(\hbar)} L_{B(\hbar)} L_{B(\hbar)} \varphi\bigl(q\exp(A(\hbar))\bigr)
    + 2 L_{B(\hbar)} L_{B'(\hbar)} \varphi\bigl(q\exp(A(\hbar))\bigr) \\
& \quad {} + L_{B'(\hbar)} L_{B(\hbar)} \varphi\bigl(q\exp(A(\hbar))\bigr)
    + L_{B''(\hbar)} \varphi\bigl(q\exp(A(\hbar))\bigr).
\end{split}
\end{equation}
At $\hbar = 0$ the function $B(\hbar)$ and its derivatives are equal
\begin{equation}
B(0) = A'(0), \quad B'(0) = A''(0), \quad
B''(0) = A'''(0) - \frac{1}{2} [A'(0),A''(0)].
\end{equation}
Therefore,
\begin{equation}
\begin{split}
\frac{\dd}{\dd\hbar}\varphi\bigl(q\exp(A(\hbar))\bigr)\bigg|_{\hbar = 0} & =
    L_{A'(0)} \varphi(q), \\
\frac{\dd^2}{\dd\hbar^2}\varphi\bigl(q\exp(A(\hbar))\bigr)\bigg|_{\hbar = 0} & =
    L_{A'(0)} L_{A'(0)} \varphi(q) + L_{A''(0)} \varphi(q), \\
\frac{\dd^3}{\dd\hbar^3}\varphi\bigl(q\exp(A(\hbar))\bigr)\bigg|_{\hbar = 0} & =
    L_{A'(0)} L_{A'(0)} L_{A'(0)} \varphi(q) + 3 L_{A'(0)} L_{A''(0)} \varphi(q)
    \\
& \quad {} - \frac{3}{2} L_{[A'(0),A''(0)]} \varphi(q) + L_{A'''(0)} \varphi(q).
\end{split}
\end{equation}
For $A(\hbar) = \hbar\bigl((-\tfrac{1}{2}(X \diamond_\hbar Y)) \diamond_\hbar
(\tfrac{1}{2}X)\bigr)$ with the help of the expansion \eqref{eq:17} we can
calculate that
\begin{equation}
A'(0) = -\frac{1}{2}Y, \quad A''(0) = -\frac{1}{4}[X,Y], \quad
A'''(0) = -\frac{3}{16}[Y,[Y,X]].
\end{equation}
Hence,
\begin{equation}
\begin{split}
\frac{\dd}{\dd\hbar}\tilde{f}\bigl(q\exp(-\tfrac{1}{2}\hbar(X \diamond_\hbar Y))
    \exp(\tfrac{1}{2}\hbar X),X\bigr) \bigg|_{\hbar = 0} & =
    -\frac{1}{2} (L_Y f)^\sim(q,X), \\
\frac{\dd^2}{\dd\hbar^2}
    \tilde{f}\bigl(q\exp(-\tfrac{1}{2}\hbar(X \diamond_\hbar Y))
    \exp(\tfrac{1}{2}\hbar X),X\bigr) \bigg|_{\hbar = 0} & =
    \frac{1}{4} (L_Y L_Y f)^\sim(q,X) - \frac{1}{4} (L_{[X,Y]} f)^\sim(q,X), \\
\frac{\dd^3}{\dd\hbar^3}
    \tilde{f}\bigl(q\exp(-\tfrac{1}{2}\hbar(X \diamond_\hbar Y))
    \exp(\tfrac{1}{2}\hbar X),X\bigr) \bigg|_{\hbar = 0} & = -\frac{1}{8}
    (L_Y L_Y L_Y f)^\sim(q,X) + \frac{3}{8} (L_Y L_{[X,Y]} f)^\sim(q,X),
\end{split}
\label{eq:19}
\end{equation}
where we identify the left-invariant vector fields $L_X$, $L_Y$ and $L_{[X,Y]}$
on $G$ with vector fields on $G \times \mathfrak{g}^*$ by means of the
identification $\T_{(q,p)}(G \times \mathfrak{g}^*) = \T_q G \oplus
\mathfrak{g}^*$. Similarly we calculate that
\begin{equation}
\begin{split}
\frac{\dd}{\dd\hbar}\tilde{g}\bigl(q\exp(\tfrac{1}{2}\hbar(X \diamond_\hbar Y))
    \exp(-\tfrac{1}{2}\hbar Y),Y\bigr) \bigg|_{\hbar = 0} & =
    \frac{1}{2} (L_X g)^\sim(q,Y), \\
\frac{\dd^2}{\dd\hbar^2}
    \tilde{g}\bigl(q\exp(\tfrac{1}{2}\hbar(X \diamond_\hbar Y))
    \exp(-\tfrac{1}{2}\hbar Y),Y\bigr) \bigg|_{\hbar = 0} & =
    \frac{1}{4} (L_X L_X g)^\sim(q,Y) + \frac{1}{4} (L_{[X,Y]} g)^\sim(q,Y), \\
\frac{\dd^3}{\dd\hbar^3}
    \tilde{g}\bigl(q\exp(\tfrac{1}{2}\hbar(X \diamond_\hbar Y))
    \exp(-\tfrac{1}{2}\hbar Y),Y\bigr) \bigg|_{\hbar = 0} & = \frac{1}{8}
    (L_X L_X L_X g)^\sim(q,Y) + \frac{3}{8} (L_X L_{[X,Y]} g)^\sim(q,Y).
\end{split}
\label{eq:20}
\end{equation}

The function $F(X)$ can be written in the form
\begin{equation}
F(X) = \sqrt{\det\lambda(\ad_X)}
= \sqrt{\det\bigl(\exp(\ln\lambda(\ad_X))\bigr)}
= \exp\left(\frac{1}{2}\Tr\bigl(\ln\lambda(\ad_X)\bigr)\right).
\end{equation}
Therefore,
\begin{equation}
L(X,Y) = \exp\left(\frac{1}{2}\Tr\bigl(\ln\lambda(\ad_X) + \ln\lambda(\ad_Y)
    - \ln\lambda(\ad_{X \diamond Y})\bigr)\right).
\end{equation}
Using the above formula the derivatives of the function $L_\hbar(X,Y) =
L(\hbar X,\hbar Y)$ can be easily calculated:
\begin{equation}
\begin{gathered}
\frac{\dd}{\dd\hbar} L_\hbar(X,Y) \bigg|_{\hbar = 0} = 0, \quad
\frac{\dd^2}{\dd\hbar^2} L_\hbar(X,Y) \bigg|_{\hbar = 0} =
    -\frac{1}{12} \Tr(\ad_X \circ \ad_Y), \\
\frac{\dd^3}{\dd\hbar^3} L_\hbar(X,Y) \bigg|_{\hbar = 0} =
    -\frac{1}{8} \Tr(\ad_{X + Y} \circ \ad_{[X,Y]}).
\end{gathered}
\label{eq:21}
\end{equation}

With the help of the expansion \eqref{eq:17} of the product $\diamond_\hbar$ we
get the following derivatives of the exponent:
\begin{equation}
\begin{split}
\frac{\dd}{\dd\hbar} e^{-i\braket{p,X \diamond_\hbar Y}} \bigg|_{\hbar = 0} & =
    -\frac{1}{2}i \braket{p,[X,Y]} e^{-i\braket{p,X + Y}}, \\
\frac{\dd^2}{\dd\hbar^2} e^{-i\braket{p,X \diamond_\hbar Y}} \bigg|_{\hbar = 0}
    & = \left(-\frac{1}{4} \braket{p,[X,Y]} \braket{p,[X,Y]}
    - \frac{1}{6}i \braket{p,[X - Y,[X,Y]]}\right) e^{-i\braket{p,X + Y}}, \\
\frac{\dd^3}{\dd\hbar^3} e^{-i\braket{p,X \diamond_\hbar Y}} \bigg|_{\hbar = 0}
    & = \biggl(\frac{1}{8}i \braket{p,[X,Y]} \braket{p,[X,Y]} \braket{p,[X,Y]}
    - \frac{1}{4} \braket{p,[X,Y]} \braket{p,[X - Y,[X,Y]]} \\
& \quad {} + \frac{1}{4}i \braket{p,[Y,[X,[X,Y]]]}\biggr)e^{-i\braket{p,X + Y}}.
\end{split}
\label{eq:22}
\end{equation}

Let $X_1,\dotsc,X_n$ be a basis in $\mathfrak{g}$ and $Y_1,\dotsc,Y_n,Z^1,
\dotsc,Z^n$ a corresponding frame fields on $M$ as in Section~\ref{subsec:2.2}.
Moreover, let $p_j \colon M \to \mathbb{R}$ be fiber variables defined by
$p_j(q,p) = \braket{p,X_j}$ and $C^k_{ij}$ the structure constants of
$\mathfrak{g}$ in the basis $\{X_i\}$. After expanding $X$, $Y$ in the basis
$\{X_i\}$ and using formulas \eqref{eq:19}, \eqref{eq:20}, \eqref{eq:21}, and
\eqref{eq:22} we can calculate the derivatives with respect to $\hbar$ of the
\nbr{\star}product. Formula \eqref{eq:18} then follows.
\end{proof}

\subsection{Extension of the $\star$-product to the Hilbert space
$\mathcal{L}(M)$}
\label{subsec:3.2}
In what follows we will extend the \nbr{\star}product to the Hilbert space
$\mathcal{L}(M)$ introduced in Section~\ref{subsec:2.6}. The space
$\mathcal{F}(M)$ inherits from $\mathcal{L}(M)$ the scalar product
\eqref{eq:15}. The norm corresponding to this scalar product will be denoted by
$\norm{\sdot}_2$.

\begin{theorem}
For $f,g \in \mathcal{F}(M)$ the following inequality holds
\begin{equation}
\norm{f \star g}_2 \leq \norm{f}_2 \norm{g}_2.
\label{eq:16}
\end{equation}
\end{theorem}

\begin{proof}
Let $\mathsf{f},\mathsf{g} \in C^\infty_c(G \times G)$ correspond to $f$ and $g$
as in \eqref{eq:5}. Then, in accordance to \eqref{eq:23}, a function
$\mathsf{h} \in C^\infty_c(G \times G)$ corresponding to the product $f \star g$
will be given by the formula \eqref{eq:29}. The functions
$\mathsf{f},\mathsf{g},\mathsf{h}$ are elements of the Hilbert space
$L^2(G \times G,\dd{m} \times \dd{m})$. By virtue of Lemma~\ref{lem:1} their
\nbr{L^2}norms are equal to the \nbr{L^2}norms of $f$, $g$ and $f \star g$.
Therefore, the inequality \eqref{eq:16} is equivalent to the following
inequality
\begin{equation}
\norm{\mathsf{h}}_2 \leq \norm{\mathsf{f}}_2 \norm{\mathsf{g}}_2.
\label{eq:24}
\end{equation}

For fixed $a,b \in G$ we can apply the Schwartz inequality to functions
$\mathsf{f}(a,\sdot),\mathsf{g}(\sdot,b) \in L^2(G,\dd{m})$ receiving the
following inequality
\begin{equation}
\Abs{\int_G \mathsf{f}(a,c) \mathsf{g}(c,b) \ud{m(c)}}^2 \leq
    \int_G \abs{\mathsf{f}(a,c)}^2 \ud{m(c)}
    \int_G \abs{\mathsf{g}(d,b)}^2 \ud{m(d)}.
\end{equation}
With its use we calculate that
\begin{align}
\norm{\mathsf{h}}_2^2 & = \int_G \int_G \abs{\mathsf{h}(a,b)}^2
    \ud{m(a)}\ud{m(b)}
= \int_G \int_G \Abs{\int_G \mathsf{f}(a,c) \mathsf{g}(c,b) \ud{m(c)}}^2
    \ud{m(a)}\ud{m(b)} \nonumber \\
& \leq \int_G \int_G \left(\int_G \abs{\mathsf{f}(a,c)}^2 \ud{m(c)}
    \int_G \abs{\mathsf{g}(d,b)}^2 \ud{m(d)}\right) \ud{m(a)}\ud{m(b)}
    \nonumber \\
& = \int_G \int_G \abs{\mathsf{f}(a,c)}^2 \ud{m(a)}\ud{m(c)}
    \int_G \int_G \abs{\mathsf{g}(d,b)}^2 \ud{m(d)}\ud{m(b)}
= \norm{\mathsf{f}}_2^2 \norm{\mathsf{g}}_2^2,
\end{align}
which proves \eqref{eq:24} and consequently \eqref{eq:16}.
\end{proof}

The consequence of the inequality \eqref{eq:16} is continuity of the
\nbr{\star}product in the \nbr{L^2}norm. From this property and the fact that
$\mathcal{F}(M)$ is dense in $\mathcal{L}(M)$ we can continuously extend the
\nbr{\star}product to the whole Hilbert space $\mathcal{L}(M)$. Note, that since
complex-conjugation is continuous in the \nbr{L^2}norm it will remain an
involution for the extended \nbr{\star}product. Hence, $\mathcal{L}(M)$ is an
involutive algebra with an inner product $(\sdot,\sdot)$, which satisfies the
following properties:
\begin{enumerate}[(i)]
\item for every $f \in \mathcal{L}(M)$ the maps $g \mapsto f \star g$ and
$g \mapsto g \star f$ are bounded on $\mathcal{L}(M)$ (multiplication is a
bounded operator),
\item $(f \star g,h) = (g,\bar{f} \star h)$ and $(g \star f,h) =
(g,h \star \bar{f})$ (involution is the adjoint),
\item $(f,g) = (\bar{g},\bar{f})$ (involution is an antilinear isometry),
\item $\mathcal{L}(M) \star \mathcal{L}(M)$ is linearly dense in
$\mathcal{L}(M)$.
\end{enumerate}
Therefore, $\mathcal{L}(M)$ is a Hilbert algebra. Observe, that for
$f,g \in \mathcal{F}(M)$, by virtue of the property (\ref{item:1d}) of
Theorem~\ref{thm:1}, we can write the scalar product of $f,g$ in the form
\begin{equation}
(f,g) = \int_M \bar{f} \star g \ud{l}.
\end{equation}
We can also extend the trace functional $\tr$ introduced in \eqref{eq:25} to the
whole space $\mathcal{L}^2(M) = \mathcal{L}(M) \star \mathcal{L}(M)$. It then
follows that
\begin{equation}
(f,g) = \tr(\bar{f} \star g), \quad f,g \in \mathcal{L}(M).
\end{equation}

\subsection{$C^*$-algebra of observables and states}
\label{subsec:3.3}
In Section~\ref{subsec:3.2} we presented an extension of the algebra
$\mathcal{F}(M)$ to the Hilbert algebra $\mathcal{L}(M)$. In what follows we
will extend $\mathcal{L}(M)$ to a \nbr{C^*}algebra. Let us introduce on the
space $\mathcal{L}(M)$ the following norm:
\begin{equation}
\norm{f} = \sup\{\norm{f \star h}_2 \mid h \in \mathcal{F}(M),\ 
    \norm{h}_2 = 1\}.
\end{equation}
This is a \nbr{C^*}norm, i.e.\ it satisfies
\begin{enumerate}[(i)]
\item $\norm{f \star g} \leq \norm{f}\norm{g}$,
\item $\norm{\bar{f}} = \norm{f}$,
\item $\norm{\bar{f} \star f} = \norm{f}^2$,
\end{enumerate}
for $f,g \in \mathcal{L}(M)$. Indeed, this follows directly from the fact that
the map $h \mapsto f \star h$ is a bounded linear operator on the Hilbert space
$\mathcal{L}(M)$ defined on a dense domain $\mathcal{F}(M)$, and the norm
$\norm{f}$ is just the operator norm of this operator.

Since $\norm{f}$ is the smallest constant $C$ satisfying the inequality
\begin{equation}
\norm{f \star h}_2 \leq C\norm{h}_2 \text{ for all $h \in \mathcal{F}(M)$},
\end{equation}
it is clear from \eqref{eq:16} that $\norm{f} \leq \norm{f}_2$, and so
convergence in \nbr{L^2}norm implies convergence in the norm $\norm{\sdot}$.

The spaces $\mathcal{F}(M)$ and $\mathcal{L}(M)$ are not complete with respect
to the \nbr{C^*}norm $\norm{\sdot}$, thus $\mathcal{F}(M)$ and $\mathcal{L}(M)$
are only pre-$C^*$-algebras. However, they can be completed to a
\nbr{C^*}algebra. This completion will be denoted by $\mathcal{A}(M)$. The
algebra $\mathcal{A}(M)$ is a \nbr{C^*}algebra of observables.

In what follows we will explicitly denote the dependence of the algebra
$\mathcal{A}(M)$ on $\hbar$ by writing $\mathcal{A}_\hbar(M)$. Moreover, we will
denote by $\mathcal{A}_0(M)$ the commutative \nbr{C^*}algebra of
\nbr{\mathbb{C}}valued continuous functions on $M$ which vanish at infinity with
the usual supremum norm $\norm{f}_0 = \sup_{x \in M} \abs{f(x)}$. The space
$\mathcal{F}_0(M)$ introduced in Section~\ref{subsec:3.1} is dense in
$\mathcal{A}_0(M)$. It happens that the field of \nbr{C^*}algebras,
$\hbar \mapsto \mathcal{A}_\hbar(M)$, is a strict deformation quantization of
the symplectic manifold $M$. In fact, there holds:

\begin{theorem}
For any $f,g \in \mathcal{F}_0(M)$ there exists $\hbar_0 > 0$ such that for
every $\hbar \in \mathbb{R}$, $\abs{\hbar} < \hbar_0$ the functions
$f,g \in \mathcal{A}_\hbar(M)$ and
\begin{enumerate}[(i)]
\item the map $\hbar \mapsto \norm{f}_\hbar$ is continuous on
$(-\hbar_0,\hbar_0)$,
\item $\norm{f \star_\hbar g - f \cdot g}_\hbar \to 0$ as $\hbar \to 0$,
\item $\norm{\lshad f,g \rshad_\hbar - \{f,g\}}_\hbar \to 0$ as $\hbar \to 0$.
\end{enumerate}
\end{theorem}

The above theorem can be proved with the help of the operator representation
introduced in Section~\ref{sec:4} and through similar considerations as in
\cite{Landsman:1993a,Landsman:1993b}.

The \nbr{C^*}algebra of observables $\mathcal{A}(M)$ can be used to define
states of the system. From definition a state is a continuous linear functional
$\Lambda \colon \mathcal{A}(M) \to \mathbb{C}$, which is positively defined and
normalized to unity, i.e.
\begin{enumerate}[(i)]
\item $\norm{\Lambda} = 1$,
\item $\Lambda(\bar{f} \star f) \geq 0$ for every $f \in \mathcal{A}(M)$.
\end{enumerate}
The set of all states is convex. Extreme points of this set are called pure
states. These are states which cannot be written as convex linear
combinations of some other states. In other words $\Lambda_\text{pure}$ is a
pure state if and only if there do not exist two different states $\Lambda_1$
and $\Lambda_2$ such that $\Lambda_\text{pure} = p\Lambda_1 + (1 - p)\Lambda_2$
for some $p \in (0,1)$.

The expectation value of an observable $f \in \mathcal{A}(M)$ in a state
$\Lambda$ is from definition equal
\begin{equation}
\braket{f}_\Lambda = \Lambda(f).
\end{equation}
If $f$ is self-adjoint, i.e.\ $\bar{f} = f$, then $\braket{f}_\Lambda \in
\mathbb{R}$.

Similarly as in classical mechanics, where states can be characterized in terms
of probabilistic distribution functions on phase space, in quantum mechanics we
can also associate with states certain phase-space functions called
quasi-probabilistic distribution functions. The next two theorems provide such
characterization. Their proofs follow directly from the operator representation
introduced in Section~\ref{sec:4}.

\begin{theorem}
\label{thm:5}
If $\rho \in \mathcal{L}(M)$ satisfies
\begin{enumerate}[(i)]
\item\label{item:2a} $\bar{\rho} = \rho$,
\item\label{item:2b} $\int_M \rho \ud{l} = 1$,
\item\label{item:2c} $\int_M \bar{f} \star f \star \rho \ud{l} \geq 0$
for every $f \in \mathcal{F}(M)$,
\end{enumerate}
then the functional
\begin{equation}
\Lambda_\rho(f) = \int_M f \star \rho \ud{l} \equiv \int_M f \cdot \rho \ud{l}
\label{eq:42}
\end{equation}
is a state. Vice verse, for every state $\Lambda$ there exists a unique function
$\rho \in \mathcal{L}(M)$ satisfying properties (\ref{item:2a})--(\ref{item:2c})
such that $\Lambda = \Lambda_\rho$.
\end{theorem}

\begin{theorem}
\label{thm:6}
A state $\Lambda_\rho$ is pure if and only if the corresponding function $\rho$
is idempotent, i.e.
\begin{equation}
\rho \star \rho = \rho.
\end{equation}
\end{theorem}

From \eqref{eq:42} the expectation value of an observable
$f \in \mathcal{F}(M)$ in a state $\Lambda_\rho$ can be written in a form
\begin{equation}
\braket{f}_\rho = \int_M f(x) \rho(x) \ud{l(x)}.
\end{equation}
In the next section we will extend $\mathcal{F}(M)$ to an algebra of
distributions $\mathcal{F}_\star(M)$. The above formula will then extend, in a
direct way, to general observables $f \in \mathcal{F}_\star(M)$ and states
$\rho \in \mathcal{F}(M)$.

\subsection{Extension of the $\star$-product to an algebra of distributions}
\label{subsec:3.4}
In what follows we will extend the \nbr{\star}product to a suitable space of
distributions. The algebra $\mathcal{F}(M)$ will play the role of the space of
test functions. The following construction is based on
\cite{Gracia-Bondia:1988,Karasev:2005}. We will assume that the group $G$ is
compact. In such case $\mathcal{F}(M)$ is a Fr\'echet algebra. The following
considerations will also hold for non-compact groups, although, then we do not
have a topology on $\mathcal{F}(M)$ so we would not be dealing with any
continuity issues.

We will denote by $\mathcal{F}'(M)$ the space of continuous linear functionals
on $\mathcal{F}(M)$, i.e.\ distributions. The dual space $\mathcal{F}'(M)$ is
endowed with the strong dual topology, that of uniform convergence on bounded
subsets of $\mathcal{F}(M)$. For $f \in \mathcal{F}'(M)$ we will denote by
$\braket{f,h}$ the value of the functional $f$ at $h \in \mathcal{F}(M)$. We
will identify functions $f \in \mathcal{F}(M)$ with the following functionals
\begin{equation}
h \mapsto \int_M f(x)h(x) \ud{x}.
\label{eq:12}
\end{equation}
These functionals are continuous on $\mathcal{F}(M)$. Indeed, if
$\mathsf{f},\mathsf{h} \in C^\infty_c(G \times G)$ are integral kernels
corresponding to $f$ and $h$, then by Lemma~\ref{lem:1}
\begin{equation}
\Abs{\int_M f(x)h(x) \ud{x}} = \abs{2\pi\hbar}^n \Abs{\int_{G \times G}
    \mathsf{f}(a,b) \mathsf{h}(b,a) \ud{m(a)}\ud{m(b)}}
\leq \abs{2\pi\hbar}^n \int_{G \times G} \abs{\mathsf{f}(a,b)}\ud{m(a)}\ud{m(b)}
    \norm{h}_{0,0},
\end{equation}
which proves the continuity of the functionals \eqref{eq:12}.
Thus we may write $\mathcal{F}(M) \subset \mathcal{F}'(M)$. By the formula
\eqref{eq:12} we can also identify other functions with distributions provided
that the above integral, treated as an improper integral from
Section~\ref{subsec:2.3}, is well defined for every test function.

For $f \in \mathcal{F}'(M)$ and $g \in \mathcal{F}(M)$ we define
$f \star g \in \mathcal{F}'(M)$ and $g \star f \in \mathcal{F}'(M)$ by the
formulas
\begin{equation}
\braket{f \star g,h} = \braket{f,g \star h}, \quad
\braket{g \star f,h} = \braket{f,h \star g} \quad
\text{for every $h \in \mathcal{F}(M)$}.
\label{eq:13}
\end{equation}
The consistency of the above definition is guaranteed by the property
(\ref{item:1d}) of Theorem~\ref{thm:1}. Note, that the maps
$g \mapsto f \star g$ and $g \mapsto g \star f$ are continuous from
$\mathcal{F}(M)$ to $\mathcal{F}'(M)$, since the maps
$g \mapsto \braket{f,g \star h}$ and $g \mapsto \braket{f,h \star g}$ are
continuous, uniformly for $h$ in a bounded subset of $\mathcal{F}(M)$
(by the joint continuity of $g$ and $h$).

Denote by $\mathcal{F}_\star(M)$ the following subspace of distributions:
\begin{equation}
\mathcal{F}_\star(M) = \{f \in \mathcal{F}'(M) \mid
    \text{$f \star g$ and $g \star f \in \mathcal{F}(M)$ for every
    $g \in \mathcal{F}(M)$}\}.
\end{equation}
In particular, $\mathcal{F}(M) \subset \mathcal{F}_\star(M)$. For
$f \in \mathcal{F}_\star(M)$ the maps $g \mapsto f \star g$ and
$g \mapsto g \star f$ are continuous from $\mathcal{F}(M)$ to $\mathcal{F}(M)$
by the closed graph theorem. Thus, for $f,g \in \mathcal{F}_\star(M)$ we may
define their \nbr{\star}product by the formula
\begin{equation}
\braket{f \star g,h} = \braket{f,g \star h} = \braket{g,h \star f} \quad
\text{for every $h \in \mathcal{F}(M)$}.
\label{eq:14}
\end{equation}
Straightforward calculations with the use of \eqref{eq:13} and \eqref{eq:14}
verify that $f \star g \in \mathcal{F}_\star(M)$ and the associativity of the
\nbr{\star}product. Note, that $1 \in \mathcal{F}_\star(M)$ and
$f \star 1 = 1 \star f = f$ for every $f \in \mathcal{F}_\star(M)$.

The involution in $\mathcal{F}(M)$ can be extended to the algebra
$\mathcal{F}_\star(M)$ in a natural way:
\begin{equation}
\braket{\bar{f},h} = \overline{\braket{f,\bar{h}}} \quad
\text{for every $h \in \mathcal{F}(M)$}
\end{equation}
and $f \in \mathcal{F}_\star(M)$. Thus, $\mathcal{F}_\star(M)$ is an involutive
algebra with unity, being a natural extension of the algebra $\mathcal{F}(M)$.

In what follows we will show that all smooth functions polynomial in fiber
variables $p_j$, i.e.\ functions of the form
\begin{equation}
f(q,p) = \sum_{l=0}^k f^{i_1 i_2 \dotsc i_l}(q) p_{i_1} p_{i_2} \dotsm p_{i_l}
\end{equation}
for $k \geq 0$ and $f^{i_1 i_2 \dotsc i_l} \in C^\infty(G)$, are in
$\mathcal{F}_\star(M)$.

\begin{theorem}
If $f(q)$ is a smooth function on $G$, then $f \in \mathcal{F}_\star(M)$.
\end{theorem}

\begin{proof}
Let $h \in \mathcal{F}(M)$. Since $f(q)\tilde{h}(q,0)$ is a smooth compactly
supported function on $G$ the integral in \eqref{eq:12} will be finite and
define a proper linear functional on $\mathcal{F}(M)$ (even if $G$ is not
compact). We will show that this functional is continuous. We calculate that
\begin{align}
\abs{\braket{f,h}} & = \Abs{\int_M f(q)h(q,p) \ud{m(q)}\ud{p}}
= \abs{2\pi\hbar}^n \Abs{\int_G f(q)\tilde{h}(q,0) \ud{m(q)}}
= \abs{2\pi\hbar}^n \Abs{\int_G f(q)\mathsf{h}(q,q) \ud{m(q)}} \nonumber \\
& \leq \abs{2\pi\hbar}^n \int_G \abs{f(q)} \ud{m(q)} \norm{h}_{0,0},
\end{align}
where $\mathsf{h} \in C^\infty_c(G \times G)$ is an integral kernel
corresponding to $h$. The integral in the last term is finite because $G$ is
compact. Therefore, the functional $\braket{f,\sdot}$ will be continuous.

Now, we will show that $f \in \mathcal{F}_\star(M)$, also in the case when $G$
is not compact. For $g,h \in \mathcal{F}(M)$ we have
\begin{align}
\braket{f \star g,h} & = \braket{f,g \star h} = \int_G \left(
    \int_{\mathfrak{g}^*} f(q) (g \star h)(q,p) \ud{p} \right) \ud{m(q)}
= \abs{2\pi\hbar}^n \int_G f(q) (\tilde{g} \odot \tilde{h})(q,0) \ud{m(q)}
    \nonumber \\
& = \abs{2\pi\hbar}^n \int_G \left( \int_{\mathfrak{g}} f(q)
    \tilde{g}\bigl(q\exp(\tfrac{1}{2}Y),Y\bigr)
    \tilde{h}\bigl(q\exp\bigl(\tfrac{1}{2}Y\bigr),-Y\bigr)
    \det \Ad_{\exp(-\frac{1}{2}Y)} \ud{Y} \right) \ud{m(q)} \nonumber \\
& = \abs{2\pi\hbar}^n \int_{\mathfrak{g}} \left( \int_G
    f\bigl(q\exp(-\tfrac{1}{2}Y)\bigr) \tilde{g}(q,Y)
    \tilde{h}(q,-Y) \ud{m(q)} \right) \ud{Y} \nonumber \\
& = \int_{G \times \mathfrak{g}^*} \left( \int_{\mathfrak{g}}
    f\bigl(q\exp(-\tfrac{1}{2}Y)\bigr) \tilde{g}(q,Y)
    e^{-\frac{i}{\hbar}\braket{p,Y}} \ud{Y} \right) h(q,p)
    \ud{m(q)}\ud{p}.
\end{align}
Hence
\begin{equation}
(f \star g)(q,p) = \int_{\mathfrak{g}} f\bigl(q\exp(-\tfrac{1}{2}X)\bigr)
    \tilde{g}(q,X) e^{-\frac{i}{\hbar}\braket{p,X}} \ud{X}.
\label{eq:40}
\end{equation}
If $\mathsf{g} \in C^\infty_c(G \times G)$ corresponds to $\tilde{g}$ as in
\eqref{eq:5}, then
\begin{equation}
(f \star g)^\sim\bigl(a\exp(\tfrac{1}{2}V_a(b)),V_a(b)\bigr) F(V_a(b))^{-1} =
    f(a) \mathsf{g}(a,b)
\end{equation}
is a smooth compactly supported function on $G \times G$. Thus $f \star g \in
\mathcal{F}(M)$. Similarly we can prove that $g \star f \in \mathcal{F}(M)$.
Therefore $f \in \mathcal{F}_\star(M)$.
\end{proof}

\begin{theorem}
\label{thm:2}
If $p_j$ is a fiber variable corresponding to a basis $\{X_i\}$ in
$\mathfrak{g}$, then $p_j \in \mathcal{F}_\star(M)$.
\end{theorem}

\begin{proof}
Let $h \in \mathcal{F}(M)$. Since for a fixed $q \in G$ the function
$X \mapsto \tilde{h}(q,X)$ is smooth in the neighborhood of 0, we have
\begin{equation}
\braket{p_j,h} = \int_G \left(
    \int_{\mathfrak{g}^*} p_j(q,p) h(q,p) \ud{p} \right) \ud{m(q)}
= -i\hbar \abs{2\pi\hbar}^n \int_G \frac{\partial}{\partial X_j}
    \tilde{h}(q,X) \bigg|_{X = 0} \ud{m(q)}.
\end{equation}
The function $\tfrac{\partial}{\partial X_j} \tilde{h}(q,X) \big|_{X = 0}$ is
smooth and compactly supported on $G$, therefore its integral will be finite.
Thus, the functional $\braket{p_j,\sdot}$ is well defined (even when $G$ is not
compact). We will show that this functional is continuous. Let
$\mathsf{h} \in C^\infty_c(G \times G)$ correspond to $\tilde{h}$ as in
\eqref{eq:5}. Then we calculate that
\begin{align}
\braket{p_j,h} & = i\hbar \abs{2\pi\hbar}^n \frac{\partial}{\partial X_j} \int_G
    \tilde{h}(q,-X) \ud{m(q)} \bigg|_{X = 0} \nonumber \\
& = i\hbar\abs{2\pi\hbar}^n \frac{\partial}{\partial X_j}
    \int_G \mathsf{h}\bigl(q\exp(\tfrac{1}{2}X),q\exp(-\tfrac{1}{2}X)\bigr)
    F(X) \ud{m(q)} \bigg|_{X = 0} \nonumber \\
& = i\hbar \abs{2\pi\hbar}^n \frac{\partial}{\partial X_j} \int_G
    \mathsf{h}(q\exp(X),q) F(X) \det\Ad_{\exp(-\frac{1}{2}X)} \ud{m(q)}
    \bigg|_{X = 0}.
\label{eq:31}
\end{align}
For $\varphi \in C^\infty(G)$ we have that
\begin{equation}
\frac{\partial}{\partial X_j} \varphi\bigl(q\exp(X)\bigr) \bigg|_{X = 0} =
    \T_q \varphi\bigl(\T_e L_q X_j\bigr) = \T_q \varphi\bigl(L_{X_j}(q)\bigr)
= L_{X_j} \varphi(q),
\label{eq:27}
\end{equation}
where $L_{X_j}$ is a left-invariant vector field corresponding to $X_j$. If
$X = u^i X_i$ is an expansion of $X$ in the basis $\{X_i\}$ and $\{X^i\}$ is
a dual basis to $\{X_i\}$, then
\begin{equation}
\Tr(\ad_X) = u^i \Tr(\ad_{X_i}) = u^i \braket{X^j,\ad_{X_i} X_j}
= u^i \braket{X^j,C^k_{ij} X_k} = C^k_{ij} u^i \delta^j_k = C^k_{ik} u^i.
\end{equation}
Hence
\begin{equation}
\frac{\partial}{\partial X_j} \det\Ad_{\exp(-\frac{1}{2}X)} \bigg|_{X = 0} =
    \frac{\partial}{\partial X_j} e^{-\frac{1}{2}\Tr(\ad_X)} \bigg|_{X = 0}
= \frac{\partial}{\partial u^j} e^{-\frac{1}{2}C^k_{ik} u^i} \bigg|_{u = 0}
= -\frac{1}{2} C^k_{jk}.
\label{eq:28}
\end{equation}
Since $F'(0) = 0$ we get, with the use of \eqref{eq:27} and \eqref{eq:28}, that
\begin{equation}
\frac{\partial}{\partial X_j} \mathsf{h}(q\exp(X),q) F(X)
    \det\Ad_{\exp(-\frac{1}{2}X)} \bigg|_{X = 0} = L_{X_j}^{(1)} \mathsf{h}(q,q)
    - \frac{1}{2}C^k_{jk} \mathsf{h}(q,q),
\label{eq:32}
\end{equation}
where $L_{X_j}^{(1)}$ is a vector fields on $G \times G$ defined as in
\eqref{eq:30}. Because compact groups are unimodular $C^k_{jk} = 0$ and we get
from \eqref{eq:31} and \eqref{eq:32} that
\begin{equation}
\abs{\braket{p_j,h}} \leq (2\pi)^n \abs{\hbar}^{n+1} m(G) \norm{h}_{k,0}
\end{equation}
for $k = (0,\dotsc,0,1,0,\dotsc,0)$ where there is 1 on the \nbr{j}th place.
This shows the continuity of the functional $\braket{p_j,\sdot}$.

Now, we will show that $p_j \in \mathcal{F}_\star(M)$, also in the case when $G$
is not compact. For $g,h \in \mathcal{F}(M)$ we have
\begin{align}
\braket{p_j \star g,h} & = \braket{p_j,g \star h} = \int_G \left(
    \int_{\mathfrak{g}^*} p_j(q,p) (g \star h)(q,p) \ud{p} \right) \ud{m(q)}
    \nonumber \\
& = \abs{2\pi\hbar}^n \int_G i\hbar \frac{\partial}{\partial X_j}
    (\tilde{g} \odot \tilde{h})(q,-X) \bigg|_{X = 0} \ud{m(q)}.
\end{align}
Let $\mathsf{g},\mathsf{h} \in C^\infty_c(G \times G)$ correspond to $\tilde{g}$
and $\tilde{h}$ as in \eqref{eq:5}. Then we can write
\begin{align}
\braket{p_j \star g,h} & = i\hbar\abs{2\pi\hbar}^n \frac{\partial}{\partial X_j}
    \int_G \int_G \mathsf{g}\bigl(q\exp(\tfrac{1}{2}X),c\bigr)
    \mathsf{h}\bigl(c,q\exp(-\tfrac{1}{2}X)\bigr) F(X) \ud{m(c)}\ud{m(q)}
    \bigg|_{X = 0} \nonumber \\
& = i\hbar \abs{2\pi\hbar}^n \frac{\partial}{\partial X_j}
    \int_G \int_G \mathsf{g}(q\exp(X),c) \mathsf{h}(c,q) F(X)
    \det\Ad_{\exp(-\frac{1}{2}X)} \ud{m(c)}\ud{m(q)} \bigg|_{X = 0} \nonumber \\
& = \abs{2\pi\hbar}^n \int_G \int_G \mathsf{f}(q,c) \mathsf{h}(c,q)
    \ud{m(c)}\ud{m(q)},
\end{align}
where in accordance to \eqref{eq:32}
\begin{equation}
\mathsf{f}(q,c) = i\hbar\frac{\partial}{\partial X_j} \mathsf{g}(q\exp(X),c)
    F(X) \det\Ad_{\exp(-\frac{1}{2}X)} \bigg|_{X = 0}
= i\hbar\left(L_{X_j}^{(1)} \mathsf{g}(q,c)
    - \frac{1}{2}C^k_{jk}\mathsf{g}(q,c)\right),
\end{equation}
so that $\mathsf{f} \in C^\infty_c(G \times G)$. By Lemma~\ref{lem:1} we get
that
\begin{equation}
\braket{p_j \star g,h} = \int_M f(x)h(x) \ud{x},
\end{equation}
where $f \in \mathcal{F}(M)$ is a function corresponding to $\mathsf{f}$.
Therefore, $f = p_j \star g$ which shows that $p_j \star g \in \mathcal{F}(M)$.
Analogically we can prove that $g \star p_j \in \mathcal{F}(M)$. Thus
$p_j \in \mathcal{F}_\star(M)$.
\end{proof}

\begin{theorem}
If $f(q)$ is a smooth function on $G$, then for every $g \in \mathcal{F}(M)$
the following expansion in $\hbar$ holds
\begin{equation}
f \star g = \sum_{k=0}^\infty \frac{1}{k!} \left(-\frac{i\hbar}{2}\right)^k
    L_{X_{i_1}} L_{X_{i_2}} \dotsm L_{X_{i_k}}f Z^{i_1} Z^{i_2} \dotsm Z^{i_k}g,
\label{eq:41}
\end{equation}
where $L_{X_i}$ are left-invariant vector fields corresponding to a basis
$\{X_i\}$ in $\mathfrak{g}$.
\end{theorem}

\begin{proof}
First observe that the following equality holds
\begin{equation}
\frac{\dd^k}{\dd t^k} f\bigl(q\exp(tX)\bigr) \bigg|_{t = 0} =
    \underbrace{L_X \dotsm L_X}_k f(q).
\end{equation}
Expanding function $t \mapsto f(q\exp(tX))$ in a Taylor series and using the
above equality results in
\begin{equation}
f\bigl(q\exp(tX)\bigr) = \sum_{k=0}^\infty \frac{1}{k!} \frac{\dd^k}{\dd t^k}
    f\bigl(q\exp(tX)\bigr) \bigg|_{t = 0} t^k
= \sum_{k=0}^\infty \frac{1}{k!} t^k \underbrace{L_X \dotsm L_X}_k f(q).
\end{equation}
For $X = u^i X_i$ and $t = -\tfrac{1}{2}$ this gives
\begin{equation}
f\bigl(q\exp(-\tfrac{1}{2}X)\bigr) = \sum_{k=0}^\infty \frac{1}{k!}
    \left(-\frac{1}{2}\right)^k L_{X_{i_1}} L_{X_{i_2}} \dotsm L_{X_{i_k}}f(q)
    u^{i_1} u^{i_2} \dotsm u^{i_k}.
\end{equation}
From the above result and the integral formula \eqref{eq:40} for the product
$f \star g$ we get \eqref{eq:41}.
\end{proof}

\begin{theorem}
If $p_j$ is a fiber variable corresponding to a basis $\{X_i\}$ in
$\mathfrak{g}$, then for every $g \in \mathcal{F}(M)$ the following expansion
in $\hbar$ holds
\begin{align}
& p_j \star g = p_j g + \sum_{k=1}^\infty \frac{(i\hbar)^k}{k!} \biggl(
    \frac{1 - 2^k}{2^{k-1}} B_k C^{j_2}_{i_2 j} C^{j_3}_{i_3 j_2}
    C^{j_4}_{i_4 j_3} \dotsm C^{j_k}_{i_k j_{k-1}}
    Z^{i_2} Z^{i_3} \dotsm Z^{i_k} Y_{j_k} g \nonumber \\
& \quad {} + B_k C^{j_1}_{i_1 j} C^{j_2}_{i_2 j_1} C^{j_3}_{i_3 j_2} \dotsm
    C^{j_k}_{i_k j_{k-1}} p_{j_k} Z^{i_1} Z^{i_2} \dotsm Z^{i_k} g \nonumber \\
& \quad {} - \frac{1}{2} \sum_{l=1}^{[k/2]} \binom{k}{2l} B_{k - 2l} B_{2l}
    C^{j_1}_{j_k j_{2l}} C^{j_2}_{i_2 j_1} C^{j_3}_{i_3 j_2} \dotsm
    C^{j_{2l}}_{i_{2l} j_{2l-1}} C^{j_{2l+1}}_{i_{2l+1} j}
    C^{j_{2l+2}}_{i_{2l+2} j_{2l+1}} C^{j_{2l+3}}_{i_{2l+3} j_{2l+2}} \dotsm
    C^{j_k}_{i_k j_{k-1}} Z^{i_2} Z^{i_3} \dotsm Z^{i_k} g \biggr),
\label{eq:39}
\end{align}
where $B_k$ is the \nbr{k}th Bernoulli number and $[k/2]$ denotes the nearest
integer number smaller or equal to $k/2$.
\end{theorem}

\begin{proof}
From the analysis of the proof of Theorem~\ref{thm:2} we can deduce that the
integral kernel $\mathsf{f}$ of $p_j \star g$ will be given by the formula
\begin{equation}
\mathsf{f}(a,b) = i\hbar\frac{\partial}{\partial X_j} \tilde{g}\bigl(
    a\exp(X)\exp(\tfrac{1}{2}V_{a\exp(X)}(b)),V_{a\exp(X)}(b)\bigr)
    F(V_{a\exp(X)}(b))^{-1} F(X) \det\Ad_{\exp(-\frac{1}{2}X)} \bigg|_{X = 0}.
\end{equation}
Therefore, from the equality
\begin{equation}
V_{q\exp(-\frac{1}{2}Y)\exp(X)}\bigl(q\exp(\tfrac{1}{2}Y)\bigr) =
    (-X) \diamond Y,
\end{equation}
we get
\begin{align}
(p_j \star g)(q,p) & = \int_{\mathcal{O}} i\hbar\frac{\partial}{\partial X_j}
    \tilde{g}\bigl(q\exp(-\tfrac{1}{2}Y)\exp(X)
    \exp(\tfrac{1}{2}((-X) \diamond Y)),(-X) \diamond Y\bigr)
    F((-X) \diamond Y)^{-1} F(X) \nonumber \\
& \quad \times \det\Ad_{\exp(-\frac{1}{2}X)} \bigg|_{X = 0}
    e^{-\frac{i}{\hbar}\braket{p,Y}} F(Y) \ud{Y} \nonumber \\
& = \int_{\mathfrak{g}} i\hbar\frac{\partial}{\partial X_j}
    \tilde{g}\bigl(q\exp(\tfrac{1}{2}(X \diamond Y))\exp(-\tfrac{1}{2}Y),Y\bigr)
    e^{-\frac{i}{\hbar}\braket{p,X \diamond Y}} L(X,Y) \bigg|_{X = 0} \ud{Y},
\label{eq:35}
\end{align}
where $\tfrac{\partial}{\partial X_j}$ refers to a differentiation with respect
to $X$ variable.

Observe, that from the integral version of the Baker-Campbell-Hausdorff formula
\begin{equation}
X \diamond Y = Y + \left(\int_0^1 \varphi(e^{s\ad_X} e^{\ad_Y}) \ud{s}\right)X,
\end{equation}
where $\varphi(x) = \frac{\ln x}{x - 1}$. Using this formula we can calculate
that
\begin{align}
\frac{\partial}{\partial X_j} (X \diamond Y) & =
    \frac{\dd}{\dd{t}} \bigl((X + tX_j) \diamond Y\bigr) \bigg|_{t = 0}
= \frac{\dd}{\dd{t}} \left(Y + \left(\int_0^1 \varphi(e^{s\ad_{X + tX_j}}
    e^{\ad_Y}) \ud{s}\right)(X + tX_j)\right) \bigg|_{t = 0} \nonumber \\
& = \frac{\dd}{\dd{t}} \left(Y + \left(\int_0^1 \varphi(e^{s\ad_{X + tX_j}}
    e^{\ad_Y}) \ud{s}\right)X + t\left(\int_0^1 \varphi(e^{s\ad_{X + tX_j}}
    e^{\ad_Y}) \ud{s}\right)X_j\right) \bigg|_{t = 0} \nonumber \\
& = \left(\frac{\dd}{\dd{t}} \int_0^1 \varphi(e^{s\ad_{X + tX_j}} e^{\ad_Y})
    \ud{s} \bigg|_{t = 0}\right)X
    + \left(\int_0^1 \varphi(e^{s\ad_X} e^{\ad_Y}) \ud{s}\right)X_j.
\end{align}
From this we get that
\begin{equation}
\frac{\partial}{\partial X_j} (X \diamond Y) \bigg|_{X = 0} =
    \varphi(e^{\ad_Y})X_j
= \sum_{k=0}^\infty \frac{B_k}{k!} \ad_Y^k X_j,
\label{eq:34}
\end{equation}
since $\varphi(e^x) = \frac{x}{e^x - 1} = \sum_{k=0}^\infty \frac{B_k}{k!} x^k$
is a generating function for the Bernoulli numbers.
If $Y = v^i X_i$ is an expansion of $Y$ in the basis $\{X_i\}$, then by
\eqref{eq:34} we get
\begin{equation}
\frac{\partial}{\partial X_j} e^{-\frac{i}{\hbar}\braket{p,X \diamond Y}}
    \bigg|_{X = 0} = -\frac{i}{\hbar} \sum_{k=0}^\infty \frac{B_k}{k!}
    C^{j_1}_{i_1 j} C^{j_2}_{i_2 j_1} C^{j_3}_{i_3 j_2} \dotsm
    C^{j_k}_{i_k j_{k-1}} p_{j_k} v^{i_1} v^{i_2} \dotsm v^{i_k}
    e^{-\frac{i}{\hbar}\braket{p,Y}}.
\label{eq:36}
\end{equation}
We have that
\begin{equation}
F(X) = \exp\left(\frac{1}{2}\Tr f(\ad_X)\right) \quad \text{for
$f(x) = \ln(\lambda(x)) = \ln\left(\frac{2}{x}\sinh\frac{x}{2}\right)$}
\end{equation}
and
\begin{equation}
\frac{\partial}{\partial X_j} F(X) = \frac{1}{2}F(X)
    \frac{\partial}{\partial X_j} \Tr f(\ad_X)
= \frac{1}{2} F(X) \Tr\bigl(f'(\ad_X) \circ \ad_{X_j}\bigr).
\end{equation}
Since
\begin{equation}
f'(x) = \frac{1}{2} \coth\frac{x}{2} - \frac{1}{x}
= \sum_{k=1}^\infty \frac{B_{2k} x^{2k-1}}{(2k)!}
\end{equation}
we get
\begin{align}
\frac{\partial}{\partial X_j} F(X) & = \frac{1}{2}F(X) \sum_{k=1}^\infty
    \frac{B_{2k}}{(2k)!}\Tr\bigl(\ad_X^{2k-1} \circ \ad_{X_j}\bigr) \nonumber \\
& = \frac{1}{2}F(X) \sum_{k=1}^\infty \frac{B_{2k}}{(2k)!} C^{j_1}_{j j_{2k}}
    C^{j_2}_{i_2 j_1} C^{j_3}_{i_3 j_2} \dotsm C^{j_{2k}}_{i_{2k} j_{2k-1}}
    u^{i_2} u^{i_3} \dotsm u^{i_{2k}}
\label{eq:33}
\end{align}
for $X = u^i X_i$. We calculate that
\begin{align}
\frac{\partial}{\partial X_j} L(X,Y) \bigg|_{X = 0} & =
    \frac{\partial}{\partial X_j} F(X)F(Y)F(X \diamond Y)^{-1} \bigg|_{X = 0}
= F(Y) \frac{\partial}{\partial X_j} F(X \diamond Y)^{-1} \bigg|_{X = 0}
    \nonumber \\
& = -F(Y)^{-1} \frac{\partial}{\partial X_j} F(X \diamond Y) \bigg|_{X = 0}
= -F(Y)^{-1} \frac{\partial F}{\partial X_i}(Y) \frac{\partial}{\partial X_j}
    (X \diamond Y)^i \bigg|_{X = 0},
\end{align}
where $(X \diamond Y)^i$ denotes the \nbr{i}th component of $X \diamond Y$ in
the expansion with respect to the basis $\{X_i\}$. Combining \eqref{eq:34} and
\eqref{eq:33} gives
\begin{align}
\frac{\partial}{\partial X_j} L(X,Y) \bigg|_{X = 0} & = -\frac{1}{2}
    \sum_{k=0}^\infty \sum_{l=1}^{[k/2]} \frac{B_{k - 2l}}{(k - 2l)!}
    \frac{B_{2l}}{(2l)!} C^{j_1}_{j_k j_{2l}} C^{j_2}_{i_2 j_1}
    C^{j_3}_{i_3 j_2} \dotsm C^{j_{2l}}_{i_{2l} j_{2l-1}} \nonumber \\
& \quad \times C^{j_{2l+1}}_{i_{2l+1} j}
    C^{j_{2l+2}}_{i_{2l+2} j_{2l+1}} C^{j_{2l+3}}_{i_{2l+3} j_{2l+2}} \dotsm
    C^{j_k}_{i_k j_{k-1}} v^{i_2} v^{i_3} \dotsm v^{i_k}.
\label{eq:37}
\end{align}
For $\varphi \in C^\infty(G)$ we have
\begin{align}
& \frac{\partial}{\partial X_j} \varphi\bigl(q\exp(\tfrac{1}{2}(X \diamond Y))
    \exp(-\tfrac{1}{2}Y)\bigr) \bigg|_{X = 0} =
    \frac{\partial}{\partial X_j} \varphi\bigl(q\exp(-\tfrac{1}{2}Y)
    \exp(\tfrac{1}{2}Y)\exp(\tfrac{1}{2}(X \diamond Y))
    \exp(-\tfrac{1}{2}Y)\bigr) \bigg|_{X = 0} \nonumber \\
& \quad = \frac{\partial}{\partial X_j}
    \varphi\left(q\exp\left(-\frac{1}{2}Y\right)
    \exp\biggl(\frac{1}{2}\sum_{k=0}^\infty \frac{1}{k!}
    \left(\frac{1}{2}\right)^k \ad_Y^k (X \diamond Y)\biggr)\right)
    \Bigg|_{X = 0}.
\end{align}
By \eqref{eq:1a} and \eqref{eq:34} we get
\begin{align}
& \frac{\partial}{\partial X_j} \varphi\bigl(q\exp(\tfrac{1}{2}(X \diamond Y))
    \exp(-\tfrac{1}{2}Y)\bigr) \bigg|_{X = 0} =
\T_q \varphi\left(\T_e L_q \biggl(\phi(\ad_{\frac{1}{2}Y}) \frac{1}{2}
    \sum_{k=0}^\infty \sum_{l=0}^\infty \frac{1}{k!}
    \left(\frac{1}{2}\right)^k \frac{B_l}{l!} \ad_Y^{k+l} X_j\biggr)\right)
    \nonumber \\
& \quad = \T_q \varphi\left(\T_e L_q \biggl(\phi(\ad_{\frac{1}{2}Y}) \frac{1}{2}
    \sum_{r=0}^\infty \frac{1}{r!} B_r(\tfrac{1}{2}) \ad_Y^r X_j\biggr)\right),
\end{align}
where
\begin{equation}
B_r(\tfrac{1}{2}) = \sum_{l=0}^r \binom{r}{l}
    \left(\frac{1}{2}\right)^{r-l} B_l = \left(\frac{1}{2^{r-1}} - 1\right) B_r
\end{equation}
is the value of the Bernoulli polynomial at $\tfrac{1}{2}$. An expansion of the
function $\phi$ in a Taylor series results in
\begin{align}
& \frac{\partial}{\partial X_j} \varphi\bigl(q\exp(\tfrac{1}{2}(X \diamond Y))
    \exp(-\tfrac{1}{2}Y)\bigr) \bigg|_{X = 0} = \nonumber \\
& \quad = -\sum_{k=0}^\infty
    \frac{1}{(k + 1)!} \sum_{r=0}^k \binom{k+1}{r}
    \left(-\frac{1}{2}\right)^{k+1-r} B_r(\tfrac{1}{2}) 
    \T_q \varphi\bigl(\T_e L_q(\ad_Y^k X_j)\bigr) \nonumber \\
& \quad = \sum_{k=1}^\infty \frac{1}{k!} \bigl(B_k(\tfrac{1}{2}) - B_k(0)\bigr)
    \T_q \varphi\bigl(\T_e L_q(\ad_Y^{k-1} X_j)\bigr),
\end{align}
where we have used the translation property of Bernoulli polynomials. Using the
above result we get
\begin{align}
& \frac{\partial}{\partial X_j}
    \tilde{g}\bigl(q\exp(\tfrac{1}{2}(X \diamond Y))\exp(-\tfrac{1}{2}Y),Y\bigr)
    \bigg|_{X = 0} = \nonumber \\
& \quad = \sum_{k=1}^\infty \frac{1}{k!} \frac{1 - 2^k}{2^{k-1}} B_k
    C^{j_2}_{i_2 j} C^{j_3}_{i_3 j_2} C^{j_4}_{i_4 j_3} \dotsm
    C^{j_k}_{i_k j_{k-1}} (Y_{j_k} g)^\sim(q,Y) v^{i_2} v^{i_3} \dotsm v^{i_k}.
\label{eq:38}
\end{align}
Applying \eqref{eq:36}, \eqref{eq:37} and \eqref{eq:38} to \eqref{eq:35} we
receive \eqref{eq:39}.
\end{proof}

Now we can show that all smooth functions polynomial in fiber variables $p_j$
belong to $\mathcal{F}_\star(M)$. First we will prove that polynomials in $p_j$
are in $\mathcal{F}_\star(M)$. Assume that for a given $k \geq 1$ all monomials
$p_{i_1} p_{i_2} \dotsm p_{i_k}$ of order $k$ are in $\mathcal{F}_\star(M)$.
Then, the expansion \eqref{eq:39} will also hold for these monomials in place of
$g$ and we can see that the product $p_{i_{k+1}} \star p_{i_1} p_{i_2} \dotsm
p_{i_k}$ will be of the form of a monomial $p_{i_1} p_{i_2} \dotsm p_{i_k}
p_{i_{k+1}}$ plus some polynomial of order $k$. Thus monomials $p_{i_1} p_{i_2}
\dotsm p_{i_k} p_{i_{k+1}}$ of order $k + 1$ will also belong to
$\mathcal{F}_\star(M)$. Since in Theorem~\ref{thm:2} we showed that all $p_j$
are in $\mathcal{F}_\star(M)$, then by induction we see that all monomials, and
consequently all polynomials, are in $\mathcal{F}_\star(M)$. Next, assume that
for a given $k \geq 0$ all smooth functions which are polynomial of order $k$ in
$p_j$ are in $\mathcal{F}_\star(M)$. The expansion \eqref{eq:41} also holds for
$g$ replaced by an arbitrary polynomial in $p_j$. Then, we can see that the
\nbr{\star}product of an arbitrary smooth function $f(q)$ on $G$ and a monomial
$p_{i_1} p_{i_2} \dotsm p_{i_{k+1}}$ of order $k + 1$ will be in the form of
a function $f(q) p_{i_1} p_{i_2} \dotsm p_{i_{k+1}}$ plus some smooth function
polynomial in $p_j$ of order $k$. Thus smooth functions polynomial in $p_j$ of
order $k + 1$, and by induction of all orders, will be in
$\mathcal{F}_\star(M)$.

Note, that from \eqref{eq:39} we get in particular that
\begin{equation}
p_i \star p_j = p_i p_j + \frac{1}{2}i\hbar C^k_{ij} p_k
    + \frac{1}{24} \hbar^2 C^k_{il} C^l_{jk}.
\end{equation}
From this we receive the following commutation relation
\begin{equation}
\lshad p_i,p_j \rshad = C^k_{ij} p_k
\end{equation}
being an analog of its classical counterpart \eqref{eq:48}.

\subsection{Time evolution}
\label{subsec:3.5}
For completeness of the quantization procedure we present a short description
of the time evolution of a quantum system. The time evolution is governed by
a Hamiltonian function $H \in \mathcal{F}_\star(M)$ which is, similarly as in
classical mechanics, some distinguished observable. The equation describing the
time evolution of an observable $A \in \mathcal{F}_\star(M)$ can be received
from its classical counterpart by replacing the Poisson bracket
$\{\sdot,\sdot\}$ with its deformation $\lshad\sdot,\sdot\rshad$:
\begin{equation}
\frac{\dd A}{\dd t}(t) - \lshad A(t),H \rshad = 0, \quad A(0) = A.
\end{equation}
If $A \in \mathcal{A}(M)$ and $A(t)$ is its time development, then the time
evolution of a state $\Lambda$ can be given by the formula
\begin{equation}
\Lambda(t)(A) = \Lambda(A(t)).
\end{equation}
In particular, if $\Lambda = \Lambda_\rho$ then we receive the following formula
for the time development of the pseudo-probabilistic distribution function
$\rho$:
\begin{equation}
\frac{\partial \rho}{\partial t}(t) - \lshad H,\rho(t) \rshad = 0, \quad
\rho(0) = \rho.
\end{equation}

\section{Operator representation}
\label{sec:4}
\subsection{Representation of the algebra of observables $\mathcal{A}(M)$}
\label{subsec:4.1}
By virtue of Gelfand-Naimark theorem the \nbr{C^*}algebra of observables
$\mathcal{A}(M)$ can be isometrically represented as a subalgebra of the
\nbr{C^*}algebra $\mathcal{B}(\mathcal{H})$ of bounded linear operators on a
certain Hilbert space $\mathcal{H}$. In what follows we will present an explicit
construction of this representation for $\mathcal{H} = L^2(G,\dd{m})$. We will
in fact receive a position representation of a quantum system, where the Hilbert
space $\mathcal{H}$ will play the role of the space of wave functions.

Let $\mathcal{H} = L^2(G,\dd{m})$ be a Hilbert space of \nbr{\mathbb{C}}valued
square integrable functions on $G$ with a scalar product given by
\begin{equation}
(\varphi,\psi) = \int_G \overline{\varphi(q)} \psi(q) \ud{m(q)},
\end{equation}
where $\dd{m}$ is the left-invariant Haar measure on $G$. For a function
$f \in \mathcal{F}(M)$ we define an operator $\hat{f}$ acting in $\mathcal{H}$
as an integral operator given by an integral kernel $\mathsf{f}$ corresponding
to $f$ according to the formula \eqref{eq:5}:
\begin{equation}
\hat{f}\psi(a) = \int_G \mathsf{f}(a,b) \psi(b) \ud{m(b)}.
\label{eq:43}
\end{equation}
Using \eqref{eq:5} and performing the following change of variables under the
integral sign: $b \mapsto X = V_a(b)$, formula \eqref{eq:43} takes the form
\begin{equation}
\hat{f}\psi(q) = \int_{\mathfrak{g}} \tilde{f}
    \left(q\exp(\tfrac{1}{2}X), X\right) \psi\bigl(q\exp(X)\bigr)
    F(X) \det\Ad_{\exp(-\frac{1}{2}X)} \ud{X}.
\end{equation}

\begin{theorem}
\label{thm:3}
For $f,g \in \mathcal{F}(M)$
\begin{enumerate}[(i)]
\item\label{item:3a} the map $f \mapsto \hat{f}$ is a linear isomorphism of
$\mathcal{F}(M)$ onto the space of integral operators whose integral kernels
$\mathsf{f} \in C^\infty_c(G \times G)$,
\item\label{item:3b} $\widehat{f \star g} = \hat{f} \hat{g}$,
\item\label{item:3c} $\hat{\bar{f}} = \hat{f}^\dagger$,
\item\label{item:3d} $\hat{f}$ is a trace class operator and
$\Tr(\hat{f}) = \tr(f)$,
\item\label{item:3e} the Hilbert-Schmidt scalar product of operators $\hat{f}$
and $\hat{g}$ is equal $(\hat{f},\hat{g}) \equiv \Tr(\hat{f}^\dagger \hat{g})
= (f,g)$,
\item\label{item:3f} $\norm{\hat{f}} = \norm{f}$.
\end{enumerate}
\end{theorem}

\begin{proof}
(\ref{item:3a}) This property follows immediately from the definition of the
space $\mathcal{F}(M)$.

(\ref{item:3b}) From the proof of Theorem~\ref{thm:1} it follows that the
integral kernel corresponding to the product $f \star g$ is exactly equal to
the integral kernel of $\hat{f}\hat{g}$.

(\ref{item:3c}) Let $\mathsf{f}(a,b)$ be an integral kernel corresponding to
$f$. Then the integral kernel of $\hat{f}^\dagger$ will be equal
\begin{align}
\overline{\mathsf{f}(b,a)} & =
    \overline{\tilde{f}\bigl(b\exp(\tfrac{1}{2}V_b(a)),V_b(a)\bigr)
    F(V_b(a))^{-1}}
= \tilde{\bar{f}}\bigl(b\exp(V_b(a))\exp(-\tfrac{1}{2}V_b(a)),-V_b(a)\bigr)
    F(V_b(a))^{-1} \nonumber \\
& = \tilde{\bar{f}}\bigl(a\exp(\tfrac{1}{2}V_a(b)),V_a(b)\bigr) F(V_a(b))^{-1},
\end{align}
which is just the integral kernel corresponding to $\bar{f}$.

(\ref{item:3d}) Since $\hat{f}$ is an integral operator whose integral kernel is
smooth and compactly supported it will be of trace class and its trace will be
expressed by the formula
\begin{align}
\Tr(\hat{f}) & = \int_G \mathsf{f}(q,q) \ud{m(q)}
= \int_G \tilde{f}(q,0) \ud{m(q)}
= \frac{1}{\abs{2\pi\hbar}^n} \int_G \int_{\mathfrak{g}^*} f(q,p)\ud{m(q)}\ud{p}
= \frac{1}{\abs{2\pi\hbar}^n} \int_M f(x) \ud{x} \nonumber \\
& = \tr(f).
\end{align}

(\ref{item:3e}) This property follows immediately from Lemma~\ref{lem:1}.

(\ref{item:3f}) The operator norm of a bounded operator $\hat{f}$ can be
expressed in terms of the Hilbert-Schmidt norm $\norm{\sdot}_2$ according to the
formula
\begin{equation}
\norm{\hat{f}} = \sup\{\norm{\hat{f}\hat{g}}_2 \mid
\text{$\hat{g}$ is a Hilbert-Schmidt operator and $\norm{\hat{g}}_2 = 1$}\}.
\end{equation}
Using properties (\ref{item:3b}) and (\ref{item:3e}) and the fact that
$\mathcal{F}(M)$ is dense in $\mathcal{L}(M)$ we get the result.
\end{proof}

The above theorem states that the map $f \mapsto \hat{f}$ is a faithful
\nbr{*}representation of the algebra $\mathcal{F}(M)$ on the Hilbert space
$\mathcal{H}$. From property (\ref{item:3e}) this representation can be extended
to the algebra $\mathcal{L}(M)$ resulting in a Hilbert algebra isomorphism
of $\mathcal{L}(M)$ onto the space of Hilbert-Schmidt operators
$\mathcal{B}_2(\mathcal{H})$. Moreover, from property (\ref{item:3f}) we can
further extend this representation to a representation of the \nbr{C^*}algebra
$\mathcal{A}(M)$, which will give us a \nbr{C^*}algebra isomorphism of
$\mathcal{A}(M)$ onto the \nbr{C^*}algebra of compact operators
$\mathcal{K}(\mathcal{H})$.

\subsection{Wigner functions}
\label{subsec:4.2}
In what follows we will introduce generalized Wigner functions which can be used
to characterize states of the quantum system. For $\varphi,\psi \in \mathcal{H}$
we define a Wigner function $\mathcal{W}(\varphi,\psi) \in \mathcal{L}(M)$ as a
phase-space function corresponding to the operator
$\hat{f} = \psi(\varphi,\sdot)$. The integral kernel of this operator is equal
$\mathsf{f}(a,b) = \psi(a)\overline{\varphi(b)}$, thus by \eqref{eq:44} the
Wigner function $\mathcal{W}(\varphi,\psi)$ will be expressed by the formula
\begin{equation}
\mathcal{W}(\varphi,\psi)(q,p) = \int_{\mathcal{O}}
    \overline{\varphi\bigl(q\exp(\tfrac{1}{2}X)\bigr)}
    \psi\bigl(q\exp(-\tfrac{1}{2}X)\bigr) e^{-\frac{i}{\hbar}\braket{p,X}}
    F(X) \ud{X}.
\label{eq:45}
\end{equation}
The following properties of the functions $\mathcal{W}(\varphi,\psi)$ are an
immediate consequence of Theorem~\ref{thm:3} and formula \eqref{eq:45}.

\begin{theorem}
\label{thm:4}
For $\varphi,\psi,\phi,\chi \in \mathcal{H}$ and
$f \in \mathcal{F}(M)$
\begin{enumerate}[(i)]
\item\label{item:4a} $\overline{\mathcal{W}(\varphi,\psi)} =
\mathcal{W}(\psi,\varphi)$,
\item\label{item:4b} $\int_M \mathcal{W}(\varphi,\psi) \ud{l} =
(\varphi,\psi)$,
\item\label{item:4c} $\bigl(\mathcal{W}(\varphi,\psi),\mathcal{W}(\phi,\chi)
\bigr) = \overline{(\varphi,\phi)}(\psi,\chi)$,
\item\label{item:4d} $\mathcal{W}(\varphi,\psi) \star \mathcal{W}(\phi,\chi)
= (\varphi,\chi) \mathcal{W}(\phi,\psi)$,
\item\label{item:4e} $f \star \mathcal{W}(\varphi,\psi) =
\mathcal{W}(\varphi,\hat{f}\psi)$ and $\mathcal{W}(\varphi,\psi) \star f =
\mathcal{W}(\hat{f}^\dagger\varphi,\psi)$,
\item\label{item:4f} $\frac{1}{\abs{2\pi\hbar}^n} \int_{\mathfrak{g}^*}
\mathcal{W}(\varphi,\varphi)(q,p) \ud{p} = \abs{\varphi(q)}^2$.
\end{enumerate}
\end{theorem}

We can define a tensor product of the Hilbert space $\mathcal{H}$ and its dual
$\mathcal{H}^*$ in terms of the Wigner transform $\mathcal{W}$:
\begin{equation}
\varphi^* \otimes \psi = \mathcal{W}(\varphi,\psi),
\end{equation}
where $\varphi \mapsto \varphi^*$ is an anti-linear isomorphism of $\mathcal{H}$
onto $\mathcal{H}^*$ appearing in the Riesz representation theorem. The map
$\otimes \colon \mathcal{H}^* \times \mathcal{H} \to \mathcal{L}(M)$ is
clearly bilinear and from property (\ref{item:4c}) from Theorem~\ref{thm:4}
it satisfies
\begin{equation}
(\varphi^* \otimes \psi,\phi^* \otimes \chi) = (\varphi^*,\phi^*)(\psi,\chi).
\end{equation}
Moreover, since the set of generalized Wigner functions
$\mathcal{W}(\varphi,\psi)$ is linearly dense in $\mathcal{L}(M)$ the map
$\otimes$ indeed defines a tensor product of $\mathcal{H}^*$ and
$\mathcal{H}$ equal to $\mathcal{L}(M)$.

For $f \in \mathcal{F}(M)$ we can treat $f \star {}$ as an operator
on the Hilbert space $\mathcal{L}(M)$. Then, by property (\ref{item:4e}) from
Theorem~\ref{thm:4} we get that
\begin{equation}
f \star {} = \hat{1} \otimes \hat{f}.
\end{equation}

Note, that the operator representation of the algebra $\mathcal{L}(M)$ gives a
one to one correspondence between states $\rho \in \mathcal{L}(M)$ and density
operators $\hat{\rho}$, i.e.\ trace class operators satisfying
\begin{enumerate}[(i)]
\item $\hat{\rho}^\dagger = \hat{\rho}$,
\item $\Tr(\hat{\rho}) = 1$,
\item $(\varphi,\hat{\rho}\varphi) \geq 0$ for every $\varphi \in \mathcal{H}$.
\end{enumerate}
Indeed, every density operator $\hat{\rho}$ gives rise to a continuous linear
functional on the \nbr{C^*}algebra $\mathcal{K}(\mathcal{H})$ of compact
operators, which is positively defined and normalized to unity. Such functional
will be given by the formula $\hat{f} \mapsto \Tr(\hat{f}\hat{\rho})$, and every
continuous positive linear functional on $\mathcal{K}(\mathcal{H})$ normalized
to unity will be of such form for some density operator $\hat{\rho}$. From this
observation we also immediately get Theorems~\ref{thm:5} and \ref{thm:6}
characterizing states in terms of quasi-probabilistic distribution functions.

From the correspondence between states and density operators we can see that
a function $\rho \in \mathcal{L}(M)$ is a state if and only if it is in the form
\begin{equation}
\rho = \sum_k p_k \mathcal{W}(\varphi_k,\varphi_k),
\end{equation}
where $\varphi_k \in \mathcal{H}$, $\norm{\varphi_k} = 1$, $p_k \geq 0$, and
$\sum_k p_k = 1$. In particular, $\rho \in \mathcal{L}(M)$ is a pure state if
and only if
\begin{equation}
\rho = \mathcal{W}(\varphi,\varphi)
\end{equation}
for some normalized vector $\varphi \in \mathcal{H}$.

The operator representation can be extended to the space of distributions
$\mathcal{F}'(M)$ and to the algebra $\mathcal{F}_\star(M)$ in the following
manner. For $\varphi,\psi \in C^\infty_c(G)$ the Wigner function
$\mathcal{W}(\varphi,\psi) \in \mathcal{F}(M)$. Let $f \in \mathcal{F}'(M)$,
then we can define the following bilinear form
\begin{equation}
(\varphi,\psi) \mapsto \braket{f,\mathcal{W}(\varphi,\psi)}, \quad
\varphi,\psi \in C^\infty_c(G).
\label{eq:63}
\end{equation}
If this form happens to be continuous with respect to the first variable, then
it uniquely defines (possibly unbounded) operator $\hat{f}$ with a dense domain
$C^\infty_c(G)$ by the formula
\begin{equation}
(\varphi,\hat{f}\psi) = \frac{1}{\abs{2\pi\hbar}^n}
    \braket{f,\mathcal{W}(\varphi,\psi)}, \quad
\varphi,\psi \in C^\infty_c(G).
\end{equation}
In the particular case, when $f \in \mathcal{F}(M)$, this definition of the
operator $\hat{f}$ coincides with the definition \eqref{eq:43}.

In what follows we will show that for every $f \in \mathcal{F}_\star(M)$ the
bilinear form \eqref{eq:63} is continuous with respect to the first variable.
From this will immediately follow that the whole algebra $\mathcal{F}_\star(M)$
can be represented as an algebra of (possibly unbounded) operators on
$\mathcal{H}$. From properties (\ref{item:4c}) and (\ref{item:4d}) from
Theorem~\ref{thm:4} we get that
\begin{equation}
\mathcal{W}(\varphi,\psi) = \frac{1}{\norm{\psi}^2}
    \mathcal{W}(\psi,\psi) \star \mathcal{W}(\varphi,\psi), \quad
\norm{\mathcal{W}(\varphi,\psi)}_2 = \norm{\varphi} \norm{\psi}.
\end{equation}
Using these equalities and Schwartz inequality we calculate that for
$f \in \mathcal{F}_\star(M)$ and $\varphi,\psi \in C^\infty_c(G)$
\begin{align}
\frac{1}{\abs{2\pi\hbar}^n} \abs{\braket{f,\mathcal{W}(\varphi,\psi)}} & =
    \frac{1}{\abs{2\pi\hbar}^n} \frac{1}{\norm{\psi}^2}
    \abs{\braket{f,\mathcal{W}(\psi,\psi) \star \mathcal{W}(\varphi,\psi)}}
= \frac{1}{\abs{2\pi\hbar}^n} \frac{1}{\norm{\psi}^2}
    \abs{\braket{f \star \mathcal{W}(\psi,\psi),\mathcal{W}(\varphi,\psi)}}
    \nonumber \\
& = \frac{1}{\norm{\psi}^2} \Abs{\int_M (f \star \mathcal{W}(\psi,\psi))(x)
    \mathcal{W}(\varphi,\psi)(x) \ud{l(x)}}
= \frac{1}{\norm{\psi}^2} \Abs{\bigl(\overline{f \star \mathcal{W}(\psi,\psi)},
    \mathcal{W}(\varphi,\psi)\bigr)} \nonumber \\
& \leq \frac{1}{\norm{\psi}^2} \norm{f \star \mathcal{W}(\psi,\psi)}_2
    \norm{\mathcal{W}(\varphi,\psi)}_2
= \frac{\norm{f \star \mathcal{W}(\psi,\psi)}_2}{\norm{\psi}} \norm{\varphi},
\end{align}
where we have used the fact that since $f \in \mathcal{F}_\star(M)$ then
$f \star \mathcal{W}(\psi,\psi) \in \mathcal{F}(M)$. This proves continuity of
the bilinear form \eqref{eq:63} with respect to the variable $\varphi$.
Note, that the operator $\hat{f}$ corresponding to $f \in \mathcal{F}_\star(M)$
takes values in $C^\infty_c(G)$ and the property (\ref{item:4e}) from
Theorem~\ref{thm:4} still holds for $f \in \mathcal{F}_\star(M)$ and
$\varphi,\psi \in C^\infty_c(G)$.

\subsection{Examples of operators}
\label{subsec:4.3}
In what follows we will derive formulas for operators corresponding to couple
particular functions on phase space.

\begin{theorem}
Let $f \in C^\infty(G)$, then $\hat{f}$ is an operator of multiplication by
the function $f$, i.e.
\begin{equation}
\hat{f}\psi(q) = f(q)\psi(q), \quad \psi \in C^\infty_c(G).
\end{equation}
\end{theorem}

\begin{proof}
For $\varphi,\psi \in C^\infty_c(G)$
\begin{align}
\frac{1}{\abs{2\pi\hbar}^n} \braket{f,\mathcal{W}(\varphi,\psi)} & =
    \frac{1}{\abs{2\pi\hbar}^n} \int_G \left( \int_{\mathfrak{g}^*} f(q)
    \mathcal{W}(\varphi,\psi)(q,p) \ud{p} \right) \ud{m(q)}
= \int_G f(q) \widetilde{\mathcal{W}}(\varphi,\psi)(q,0) \ud{m(q)} \nonumber \\
& = \int_G f(q) \overline{\varphi(q)} \psi(q) \ud{m(q)}
= (\varphi,f\psi).
\end{align}
\end{proof}

\begin{theorem}
Let $p_j$ be a fiber variable corresponding to a basis $\{X_i\}$ in
$\mathfrak{g}$, then
\begin{equation}
\hat{p}_j = i\hbar\left(L_{X_j} - \frac{1}{2}C^k_{jk}\right),
\label{eq:49}
\end{equation}
where $L_{X_j}$ is a left-invariant vector field corresponding to $X_j$.
In particular, for unimodular group $G$
\begin{equation}
\hat{p}_j = i\hbar L_{X_j}.
\end{equation}
\end{theorem}

\begin{proof}
For $\varphi,\psi \in C^\infty_c(G)$
\begin{align}
\frac{1}{\abs{2\pi\hbar}^n} \braket{p_j,\mathcal{W}(\varphi,\psi)} & =
    \frac{1}{\abs{2\pi\hbar}^n} \int_G \left( \int_{\mathfrak{g}^*} p_j(q,p)
    \mathcal{W}(\varphi,\psi)(q,p) \ud{p} \right) \ud{m(q)} \nonumber \\
& = -i\hbar \int_G \frac{\partial}{\partial X_j}
    \widetilde{\mathcal{W}}(\varphi,\psi)(q,X) \bigg|_{X = 0} \ud{m(q)}
    \nonumber \\
& = i\hbar \frac{\partial}{\partial X_j} \int_G
    \overline{\varphi\bigl(q\exp(-\tfrac{1}{2}X)\bigr)}
    \psi\bigl(q\exp(\tfrac{1}{2}X)\bigr) F(X) \ud{m(q)} \bigg|_{X = 0}
    \nonumber \\
& = i\hbar \frac{\partial}{\partial X_j} \int_G \overline{\varphi(q)}
    \psi\bigl(q\exp(X)\bigr) F(X) \det\Ad_{\exp(-\frac{1}{2}X)} \ud{m(q)}
    \bigg|_{X = 0}.
\end{align}
Using \eqref{eq:27} and \eqref{eq:28}, and the fact that $F'(0) = 0$ we get
\begin{equation}
\frac{1}{\abs{2\pi\hbar}^n} \braket{p_j,\mathcal{W}(\varphi,\psi)} =
    \int_G \overline{\varphi(q)} i\hbar\left(L_{X_j}\psi(q)
    - \frac{1}{2}C^k_{jk}\psi(q)\right) \ud{m(q)},
\end{equation}
which proves \eqref{eq:49}.
\end{proof}

Note, that the operators $\hat{p}_j$ satisfy the following commutation relations
\begin{equation}
[\hat{p}_i,\hat{p}_j] = i\hbar C^k_{ij} \hat{p}_k.
\end{equation}
To a function $f(q,p) = f^i(q) p_i$ ($f^i \in C^\infty(G)$) linear in fiber
variables corresponds the following operator
\begin{equation}
\hat{f} = \frac{1}{2}f^i \hat{p}_i + \frac{1}{2} \hat{p}_i f^i,
\label{eq:50}
\end{equation}
and to a function $f(q,p) = f^{ij}(q) p_i p_j$ ($f^{ij} \in C^\infty(G)$ being
symmetric with respect to indices $i,j$) quadratic in fiber variables
corresponds the operator
\begin{equation}
\hat{f} = \frac{1}{4} f^{ij} \hat{p}_i \hat{p}_j
    + \frac{1}{2} \hat{p}_i f^{ij} \hat{p}_j
    + \frac{1}{4} \hat{p}_i \hat{p}_j f^{ij}
    - \frac{1}{24} \hbar^2 C^k_{il} C^l_{jk} f^{ij}.
\label{eq:51}
\end{equation}
Indeed, with the help of the expansion \eqref{eq:39} we can write
\begin{equation}
\begin{split}
f^i p_i & = \frac{1}{2}f^i \star p_i + \frac{1}{2} p_i \star f^i, \\
f^{ij} p_i p_j & = \frac{1}{4} f^{ij} \star p_i \star p_j
    + \frac{1}{2} p_i \star f^{ij} \star p_j
    + \frac{1}{4} p_i \star p_j \star f^{ij}
    - \frac{1}{24} \hbar^2 C^k_{il} C^l_{jk} f^{ij},
\end{split}
\end{equation}
from which follow \eqref{eq:50} and \eqref{eq:51}. Observe, that \eqref{eq:50}
is a symmetrically ordered function of operators $f^i$ and $\hat{p}_i$, and
\eqref{eq:51} is a symmetrically ordered function of operators $f^{ij}$ and
$\hat{p}_i$ plus an additional term $-\tfrac{1}{24} \hbar^2 C^k_{il} C^l_{jk}
f^{ij}$ which can be treated as a quantum correction to the potential.

\subsection{Quantizer}
\label{subsec:4.4}
The operator representation can be expressed in terms of a family of operators
$\{\Delta_x \mid x \in M\}$ on the Hilbert space $\mathcal{H}$ called a
quantizer. Let us define
\begin{equation}
\Delta_x \psi(a) = \frac{1}{\abs{\pi\hbar}^n} \sqrt{J_q(a)} e^{-\frac{2i}{\hbar}
    \braket{p,V_q(a)}} \psi(s_q(a)), \quad
x = (q,p),\ q \in G,\ p \in \mathfrak{g}^*,
\label{eq:47}
\end{equation}
where $s_q$ is the reflection about point $q$ defined in \eqref{eq:46} and $J_q$
was defined in Section~\ref{subsec:2.5}. The function $\Delta_x \psi(a)$ is well
defined by the formula \eqref{eq:47} for $a \in L_q(\mathcal{U})$ and, since
$G \setminus L_q(\mathcal{U})$ is of measure zero, it can be uniquely extended
to an element of $L^2(G,\dd{m})$. Therefore, the operator $\Delta_x$ is well
defined.

\begin{theorem}
The operators $\Delta_x$ are self-adjoint on domains
\begin{equation}
\mathcal{D}(\Delta_x) = \{\psi \in L^2(G,\dd{m}) \mid \int_G j_q \abs{\psi}^2
\ud{m} < \infty\},
\end{equation}
i.e.\ $\Delta_x^\dagger = \Delta_x$.
\end{theorem}

\begin{proof}
The domain $\mathcal{D}(\Delta_x)$ is a dense vector subspace of
$L^2(G,\dd{m})$. It is the biggest domain on which operator $\Delta_x$ is well
defined. Indeed, we calculate that for $\psi \in L^2(G,\dd{m})$
\begin{align}
\norm{\Delta_x\psi}^2 & = \int_G \abs{\Delta_x\psi(a)}^2 \ud{m(a)}
= \frac{1}{\abs{\pi\hbar}^{2n}} \int_G \abs{\psi(s_q(a))}^2 j_q(a)
    \frac{\dd{m(s_q)}}{\dd{m}}\bigg|_a \ud{m(a)} \nonumber \\
& = \frac{1}{\abs{\pi\hbar}^{2n}} \int_G \abs{\psi(a)}^2 j_q(a) \ud{m(a)}.
\end{align}
Thus $\Delta_x\psi \in L^2(G,\dd{m})$ if and only if
$\psi \in \mathcal{D}(\Delta_x)$.

Now we will prove that $\Delta_x$ are self-adjoint on $\mathcal{D}(\Delta_x)$.
Let $\varphi,\psi \in \mathcal{D}(\Delta_x)$. We get
\begin{equation}
(\varphi,\Delta_x\psi) = \int_G \overline{\varphi(a)} \Delta_x\psi(a) \ud{m(a)}
= \frac{1}{\abs{\pi\hbar}^n} \int_G \overline{\varphi(a)} \psi(s_q(a))
    e^{-\frac{2i}{\hbar}\braket{p,V_q(a)}} \sqrt{J_q(a)} \ud{m(a)}.
\end{equation}
Since
\begin{equation}
\frac{\dd{m(s_q)}}{\dd{m}}\bigg|_a \frac{\dd{m(s_q)}}{\dd{m}}\bigg|_{s_q(a)} = 1
\end{equation}
we get that
\begin{equation}
\sqrt{J_q(a)} = \sqrt{J_q(s_q(a))} \frac{\dd{m(s_q)}}{\dd{m}}\bigg|_a.
\end{equation}
From the above identity and the fact that $V_q(a) = -V_q(s_q(a))$ we receive
\begin{align}
(\varphi,\Delta_x\psi) & = \frac{1}{\abs{\pi\hbar}^n} \int_G
    \overline{\varphi(a)}\psi(s_q(a)) e^{\frac{2i}{\hbar}\braket{p,V_q(s_q(a))}}
    \sqrt{J_q(s_q(a))} \frac{\dd{m(s_q)}}{\dd{m}}\bigg|_a \ud{m(a)} \nonumber \\
& = \frac{1}{\abs{\pi\hbar}^n} \int_G \overline{\varphi(s_q(a))} \psi(a)
    e^{\frac{2i}{\hbar}\braket{p,V_q(a)}} \sqrt{J_q(a)} \ud{m(a)}
= (\Delta_x\varphi,\psi).
\end{align}
Hence operators $\Delta_x$ are symmetric. To prove that they are self-adjoint
we have to show that $\mathcal{D}(\Delta_x) = \mathcal{D}(\Delta_x^\dagger)$.
Clearly $\mathcal{D}(\Delta_x) \subset \mathcal{D}(\Delta_x^\dagger)$ since
$\Delta_x$ are symmetric, so let $\psi \in \mathcal{D}(\Delta_x^\dagger)$.
For every $\varphi \in \mathcal{D}(\Delta_x)$ from the definition of
adjointness: $(\Delta_x\varphi,\psi) = (\varphi,\Delta_x^\dagger \psi)$ and
previous calculations we get
\begin{equation}
\int_G \overline{\varphi(a)} \left( \frac{1}{\abs{\pi\hbar}^n} \psi(s_q(a))
    e^{-\frac{2i}{\hbar}\braket{p,V_q(a)}} \sqrt{J_q(a)}
    - \Delta_x^\dagger \psi(a) \right) \ud{m(a)} = 0.
\end{equation}
The above equality holds for every $\varphi \in \mathcal{D}(\Delta_x)$, hence
\begin{equation}
\Delta_x^\dagger \psi(a) = \frac{1}{\abs{\pi\hbar}^n} \psi(s_q(a))
    e^{-\frac{2i}{\hbar}\braket{p,V_q(a)}} \sqrt{J_q(a)}.
\end{equation}
Since $\Delta_x^\dagger \psi \in L^2(G,\dd{m})$ also the right hand side of the
above equality has to be square integrable. Thus
$\psi \in \mathcal{D}(\Delta_x)$.
\end{proof}

\begin{theorem}
Let $f \in \mathcal{F}(M)$ and $\hat{f}$ be the corresponding operator on
$\mathcal{H}$. Then
\begin{equation}
\hat{f} = \int_M f(x) \Delta_x \ud{x}, \quad
f(x) = \abs{2\pi\hbar}^n \Tr(\Delta_x \hat{f}).
\label{eq:53}
\end{equation}
\end{theorem}

\begin{proof}
For $\psi \in C^\infty_c(G)$ and $a \in G$
\begin{align}
\int_M f(x) \Delta_x \psi(a) \ud{x}
& = \frac{1}{\abs{\pi\hbar}^n} \int_M f(q,p) \sqrt{J_q(a)}
    e^{-\frac{2i}{\hbar}\braket{p,V_q(a)}} \psi(s_q(a)) \ud{x} \nonumber \\
& = \frac{1}{\abs{\pi\hbar}^n} \int_G \left(\int_{\mathfrak{g}^*} f(q,p)
    e^{-\frac{2i}{\hbar}\braket{p,V_q(a)}} \ud{p} \right) \psi(s_q(a))
    \sqrt{J_q(a)} \ud{m(q)} \nonumber \\
& = 2^n \int_{L_a(\mathcal{U})} \tilde{f}(q,-2V_q(a)) \psi(s_q(a)) \sqrt{J_q(a)}
    \ud{m(q)}.
\end{align}
Performing the following change of variables under the integral sign
$X \mapsto q = a\exp(\tfrac{1}{2}X)$ we get
\begin{equation}
\int_M f(x) \Delta_x \psi(a) \ud{x} = \int_{2\mathcal{O}}
    \tilde{f}\bigl(a\exp(\tfrac{1}{2}X),X\bigr) \psi\bigl(a\exp(X)\bigr)
    \sqrt{J_{a\exp(\frac{1}{2}X)}(a)}
    \frac{\dd m(\exp)}{\dd X}\bigg|_{\frac{1}{2}X} \ud{X}.
\end{equation}
We calculate that
\begin{align}
\sqrt{J_{a\exp(\frac{1}{2}X)}(a)} \frac{\dd m(\exp)}{\dd X}\bigg|_{\frac{1}{2}X}
& = \sqrt{\abs{\det \psi(\ad_{\frac{1}{2}X})
    \det \Ad_{\exp(-\frac{1}{2}X)}}}
    \abs{\det \phi(\ad_{\frac{1}{2}X})} \chi_{\mathcal{O}}(X) \nonumber \\
& = \sqrt{\abs{\det \phi(\ad_X)
    \det \Ad_{\exp(\frac{1}{2}X)}}} \sqrt{\det \Ad_{\exp(X)}}
    \chi_{\mathcal{O}}(X) \nonumber \\
& = F(X) \det \Ad_{\exp(-\frac{1}{2}X)} \chi_{\mathcal{O}}(X)
\end{align}
since $\psi(x)\phi^2(x) = \phi(2x)$. Thus, because of the fact that
$\tilde{f}$ has support in $G \times \overline{\mathcal{O}}$, we get
\begin{equation}
\int_M f(x) \Delta_x \psi(a) \ud{x} = \int_{\mathfrak{g}}
    \tilde{f}\bigl(a\exp(\tfrac{1}{2}X),X\bigr) \psi\bigl(a\exp(X)\bigr)
    F(X) \det \Ad_{\exp(-\frac{1}{2}X)} \ud{X} = \hat{f}\psi(a)
\end{equation}
which proves the first equality in \eqref{eq:53}.

Let $\mathsf{f} \in C^\infty_c(G \times G)$ be an integral kernel of $\hat{f}$.
Then an integral kernel of the operator $\Delta_x \hat{f}$ is equal
\begin{equation}
(a,b) \mapsto \frac{1}{\abs{\pi\hbar}^n} \sqrt{J_q(a)} e^{-\frac{2i}{\hbar}
    \braket{p,V_q(a)}} \mathsf{f}(s_q(a),b).
\end{equation}
Therefore,
\begin{equation}
\abs{2\pi\hbar}^n \Tr(\Delta_x \hat{f}) = 2^n \int_G \sqrt{J_q(a)}
    e^{-\frac{2i}{\hbar}\braket{p,V_q(a)}} \mathsf{f}(s_q(a),a) \ud{m(a)}.
\end{equation}
Performing the following change of variables under the integral sign
$X \mapsto a = q\exp(\tfrac{1}{2}X)$ we get
\begin{align}
\abs{2\pi\hbar}^n \Tr(\Delta_x \hat{f}) & = \int_{2\mathcal{O}}
    \mathsf{f}\bigl(q\exp(-\tfrac{1}{2}X),q\exp(\tfrac{1}{2}X)\bigr)
    e^{-\frac{i}{\hbar}\braket{p,X}} \sqrt{J_q\bigl(q\exp(\tfrac{1}{2}X)\bigr)}
    \frac{\dd m(\exp)}{\dd X}\bigg|_{\frac{1}{2}X} \ud{X} \nonumber \\
& = \int_{\mathcal{O}}
    \mathsf{f}\bigl(q\exp(-\tfrac{1}{2}X),q\exp(\tfrac{1}{2}X)\bigr)
    e^{-\frac{i}{\hbar}\braket{p,X}} F(X) \ud{X} = f(x),
\end{align}
since
\begin{align}
\sqrt{J_q\bigl(q\exp(\tfrac{1}{2}X)\bigr)}
    \frac{\dd m(\exp)}{\dd X}\bigg|_{\frac{1}{2}X} & =
    \sqrt{\abs{\det \psi(\ad_{\frac{1}{2}X})
    \det \Ad_{\exp(\frac{1}{2}X)}}}
    \abs{\det \phi(\ad_{\frac{1}{2}X})} \chi_{\mathcal{O}}(X) \nonumber \\
& = F(X) \chi_{\mathcal{O}}(X).
\end{align}
This proves the second equality in \eqref{eq:53}.
\end{proof}

\begin{corollary}
For $\varphi \in L^2(G,\dd{m})$ and $\psi \in \mathcal{D}(\Delta_x)$ there holds
\begin{equation}
(\varphi,\Delta_x \psi) = \frac{1}{\abs{2\pi\hbar}^n}
    \mathcal{W}(\varphi,\psi)(x).
\label{eq:52}
\end{equation}
\end{corollary}

From \eqref{eq:52} it follows that the operator corresponding to the Dirac
delta distribution $\delta_x \in \mathcal{F}'(M)$ is equal to $\Delta_x$:
\begin{equation}
\hat{\delta}_x = \Delta_x.
\end{equation}

\section{Phase space reduction}
\label{sec:5}
The Lie group $G$ acts naturally on $M$, where the action of $G$ on $M$ is given
by the formula
\begin{equation}
\Phi \colon G \times M \to M, \quad \Phi_g(q,p) = (gq,p), \quad
g \in G,\ (q,p) \in M.
\end{equation}
If $G$ is a symmetry group of the system, then it will give rise to a reduced
system. The action $\Phi$ is free and proper, and for every $g \in G$ the map
$\Phi_g$ preserves the Poisson structure on $M$. Hence, the quotient space
$M/G$ is a smooth manifold with a natural Poisson structure. This Poisson
manifold is diffeomorphic to the dual of the Lie algebra $\mathfrak{g}^*$
endowed with the following Poisson bracket
\begin{equation}
\{f,g\} = p_k C^k_{ij} Z^i f Z^j g, \quad f,g \in C^\infty(\mathfrak{g}^*).
\label{eq:55}
\end{equation}
Moreover, \nbr{G}invariant functions in $C^\infty(M)$ are exactly those
functions which depend only on momentum variables, i.e.\ they are elements of
$C^\infty(\mathfrak{g}^*)$. Therefore, the space $C^\infty(\mathfrak{g}^*)$ is
a classical Poisson algebra of the reduced system and $\mathfrak{g}^*$ is its
phase space. Note, that the Poisson bracket \eqref{eq:55} agrees with
\eqref{eq:54} for functions depending only on momentum variables.

The quantization procedure introduced in this paper respects the reduction
operation described above. For simplicity let us assume that $G$ is compact and
that the Haar measure $\dd{m}$ is normalized to unity. Denote by
$\mathcal{F}(\mathfrak{g}^*)$ the subspace of $\mathcal{F}(M)$ consisting of
functions depending only on momentum variables. Then
$\mathcal{F}(\mathfrak{g}^*)$ is a Fr\'echet subspace of $\mathcal{F}(M)$.
The \nbr{\star}product \eqref{eq:2} of functions $f,g \in
\mathcal{F}(\mathfrak{g}^*)$ takes the form
\begin{equation}
(f \star g)(p) = \int_{\mathfrak{g} \times \mathfrak{g}}
    \tilde{f}(X) \tilde{g}(Y) e^{-\frac{i}{\hbar}\braket{p,X \diamond Y}}
    L(X,Y) \ud{X} \ud{Y}.
\end{equation}
Observe, that $f \star g$ is only a function of $p$, therefore it is an element
of $\mathcal{F}(\mathfrak{g}^*)$. Thus, $\mathcal{F}(\mathfrak{g}^*)$ is a
Fr\'echet subalgebra of $\mathcal{F}(M)$ describing a reduced quantum system.
The extension of $\mathcal{F}(\mathfrak{g}^*)$ to a Hilbert algebra
$\mathcal{L}(\mathfrak{g}^*)$ is a Hilbert subspace of $\mathcal{L}(M)$
consisting of those functions in $\mathcal{L}(M)$ which depend only on momentum
variables. Denote by $\mathcal{A}(\mathfrak{g}^*)$ the extension of
$\mathcal{L}(\mathfrak{g}^*)$ to a \nbr{C^*}algebra. The algebra
$\mathcal{A}(\mathfrak{g}^*)$ is an algebra of observables of the reduced
quantum system and defines admissible states of the system. Moreover, the
operator representation restricted to the \nbr{C^*}algebra
$\mathcal{A}(\mathfrak{g}^*)$ can be reduced to an isometric
\nbr{*}representation on some Hilbert subspace of $\mathcal{H} = L^2(G,\dd{m})$.

\section{Examples}
\label{sec:6}
\subsection{The group $\mathbb{R}^n$}
\label{subsec:6.1}
The simplest example of a Lie group $G$ is an Abelian group $(\mathbb{R}^n,+)$
where the group operation is the addition of vectors. Its Lie algebra can be
identified with the vector space $\mathbb{R}^n$ equipped with a trivial Lie
bracket $[\sdot,\sdot] = 0$. This group is trivially weakly exponential for
which $\mathcal{O} = \mathbb{R}^n$ and $\mathcal{U} = \mathbb{R}^n$. The group
$\mathbb{R}^n$ can be used to describe translational degrees of freedom of a
body or a system of $N$ bodies for $n = 3N$. In this case the \nbr{\star}product
\eqref{eq:2} takes the form of a Moyal product
\begin{equation}
(f \star g)(q,p) = \int_{\mathbb{R}^{2n}} \tilde{f}(q - \tfrac{1}{2}Y,X)
    \tilde{g}(q + \tfrac{1}{2}X,Y) e^{-\frac{i}{\hbar}\braket{p,X + Y}}
    \ud{X}\ud{Y}.
\end{equation}

\subsection{The rotation group $SO(3)$}
\label{subsec:6.2}
A non-trivial example of a Lie group $G$ which is weakly exponential
is the rotation group $SO(3)$, i.e.\ the group of $3 \times 3$ real
orthogonal matrices with determinant equal 1 ($A^T A = I$, $\det A = 1$). The
Lie algebra of $SO(3)$ is denoted by $\mathfrak{so}(3)$ and consists of all real
skew-symmetric $3 \times 3$ matrices ($X^T = -X$). The general element $X$ in
$\mathfrak{so}(3)$ is of the form
\begin{equation}
X = \begin{bmatrix}
 0 & -z &  y \\
 z &  0 & -x \\
-y &  x &  0
\end{bmatrix},
\end{equation}
where $\vec{\omega} = (x,y,z)$ is the Euler vector which represents the axis of
rotation and whose length is equal to the angle of rotation. For
$X,Y \in \mathfrak{so}(3)$ the Euler vector corresponding to the Lie bracket
$[X,Y]$ is equal to the vector product $\vec{\omega}_1 \times \vec{\omega}_2$ of
the Euler vectors corresponding to $X$ and $Y$. On $\mathfrak{so}(3)$ we can
introduce a norm $\norm{X}$ of element $X$ as the length $\abs{\vec{\omega}}$ of
the corresponding Euler vector $\vec{\omega}$. Note, that $\norm{\sdot}$ is
compatible with the Lie bracket on $\mathfrak{so}(3)$: $\norm{[X,Y]} \leq
\norm{X}\norm{Y}$. In what follows we will show that the rotation group $SO(3)$
is weakly exponential.

Any rotation can be represented by a unique angle $\theta$ in the range
$0 \leq \theta \leq \pi$ (rotation angle) and a unit vector $\vec{n}$
(rotation axis). If $0 < \theta < \pi$ the vector $\vec{n}$ is also unique,
$\theta = 0$ corresponds to the identity matrix, and for $\theta = \pi$ vectors
$\pm\vec{n}$ correspond to the same rotation. Thus, as the set $\mathcal{O}$ on
which the exponential map $\exp$ is diffeomorphic we can take the set of all
matrices $X$ which corresponding Euler vectors $\vec{\omega}$ have length in
the range $0 \leq \abs{\vec{\omega}} < \pi$. In other words
$\mathcal{O} = \{X \in \mathfrak{so}(3) \mid \norm{X} < \pi\}$ is an open ball
centered at 0 and of radius $\pi$. Then $\mathcal{U} = \exp(\mathcal{O})$
consists of all rotations except ones with angle $\pi$. Note, that rotations
with angle $\pi$ are exactly those matrices $R \in SO(3)$ for which
$\Tr R = -1$. Indeed, any matrix $R \in SO(3)$ can be written in a form $e^X$
for some $X \in \mathfrak{so}(3)$. From Rodrigues' formula
\begin{equation}
R = e^X = I + \frac{\sin\theta}{\theta}X
    + 2\frac{\sin^2\frac{\theta}{2}}{\theta^2}X^2,
\end{equation}
where $\theta = \norm{X}$ is a rotation angle. Since $\Tr X = 0$ and
$\Tr X^2 = -2\norm{X}^2$ for any $X \in \mathfrak{so}(3)$, we get that
\begin{equation}
\Tr R = 3 - 4\sin^2\frac{\theta}{2}.
\end{equation}
Thus $\Tr R = -1$ if and only if $\theta = (2k + 1)\pi$, $k \in \mathbb{Z}$,
or when rotation $R$ is with angle $\pi$. Hence we can write that $\mathcal{U} =
\{R \in SO(3) \mid \Tr R \neq -1\}$.

The exponential map $\exp \colon \mathcal{O} \to \mathcal{U}$ is indeed a
diffeomorphism. To show this we have to prove that $\exp^{-1}$ is smooth.
From inverse function theorem this will happen if and only if $\T_X \exp$ is
invertible for all $X \in \mathcal{O}$. This on the other hand is equivalent
with the invertibility of $\phi(\ad_X)$. The operator $\phi(\ad_X)$ is
invertible if and only if its eigenvalues are all nonzero. The eigenvalues of
$\phi(\ad_X)$ are equal $\phi(\lambda_j)$, where $\lambda_j$ are eigenvalues of
$\ad_X$. Putting $\phi(\lambda_j) = 0$ we see that $\T_X \exp$ is invertible
precisely when $\lambda_j \neq 2k\pi i$, $k = \pm 1, \pm 2, \dotsc$. In a basis
\begin{equation}
E_1 = \begin{bmatrix}
0 & 0 &  0 \\
0 & 0 & -1 \\
0 & 1 &  0
\end{bmatrix}, \quad
E_2 = \begin{bmatrix}
 0 & 0 & 1 \\
 0 & 0 & 0 \\
-1 & 0 & 0
\end{bmatrix}, \quad
E_3 = \begin{bmatrix}
0 & -1 & 0 \\
1 &  0 & 0 \\
0 &  0 & 0
\end{bmatrix}
\end{equation}
of the Lie algebra $\mathfrak{so}(3)$ a matrix corresponding to the operator
$\ad_X$ is just equal $X$. Thus $\ad_X$ and $X$ have the same set of
eigenvalues. The eigenvalues of $X \in \mathfrak{so}(3)$ are equal
$0, i\norm{X}, -i\norm{X}$, so it is readily seen that for $X \in \mathcal{O}$
its eigenvalues will satisfy condition $\lambda_j \neq 2k\pi i$,
$k = \pm 1, \pm 2, \dotsc$. Thus we proved that $\exp^{-1}$ is smooth.

In what follows we will prove that $SO(3) \setminus \mathcal{U}$ has measure
zero. Define function $\Phi \colon SO(3) \to \mathbb{R}$ by the formula
\begin{equation}
\Phi(A) = \Tr A + 1.
\end{equation}
This function is smooth and its differential $\dd\Phi(A)$ is nonzero for every
$A \in SO(3)$. Thus the level set $\Phi^{-1}(0) = SO(3) \setminus \mathcal{U}$
will be a submanifold of $SO(3)$ of dimension $\dim SO(3) - 1$
\cite[Corollary~5.14]{Lee:2003}. Therefore, $SO(3) \setminus \mathcal{U}$
will be of measure zero \cite[Corollary~6.12]{Lee:2003}.

\subsection{The group $SL(2,\mathbb{C})$}
\label{subsec:6.3}
Yet another example of a weakly exponential Lie group is the special linear
group $SL(2,\mathbb{C})$. It is the group of $2 \times 2$ complex matrices with
determinant equal 1. The Lie algebra $\mathfrak{sl}(2,\mathbb{C})$ of
$SL(2,\mathbb{C})$ consists of all complex, traceless $2 \times 2$ matrices.
The eigenvalues $\pm\lambda$ of a matrix $X \in \mathfrak{sl}(2,\mathbb{C})$
are equal to square roots of $-\det X$ and the following equality holds:
$X^2 = \lambda^2 I$. It is then not difficult to see, that the exponential map
on $SL(2,\mathbb{C})$ takes the form
\begin{equation}
\exp(X) = \cosh\lambda I + \frac{\sinh\lambda}{\lambda} X.
\label{eq:60}
\end{equation}
The eigenvalues of $\exp(X)$ are then equal $e^{\pm\lambda}$. Let
\begin{equation}
\mathcal{O} = \{X \in \mathfrak{sl}(2,\mathbb{C}) \mid \abs{\im\lambda} < \pi
    \text{ where $\lambda$ is one of the eigenvalues of $X$}\}
\end{equation}
and
\begin{equation}
\mathcal{U} = \{A \in SL(2,\mathbb{C}) \mid \Tr A \in \mathbb{C} \setminus
    (-\infty,-2]\}.
\end{equation}

In what follows we will show that $\exp \colon \mathcal{O} \to \mathcal{U}$ is
a diffeomorphism. First observe that for $A = e^X$
\begin{equation}
\Tr A = 2\cosh\lambda,
\label{eq:61}
\end{equation}
where $\lambda$ is an eigenvalue of $X$. On the other hand eigenvalues
$\alpha_1,\alpha_2$ of $A \in SL(2,\mathbb{C})$ can be expressed by $\Tr A$
according to the formulas
\begin{equation}
\alpha_1 = \frac{\Tr A - \sqrt{(\Tr A)^2 - 4}}{2}, \quad
\alpha_2 = \frac{\Tr A + \sqrt{(\Tr A)^2 - 4}}{2}.
\label{eq:62}
\end{equation}
Note, that $\exp$ is a bijection of $\{z \in\mathbb{C} \mid \abs{\im z} < \pi\}$
onto $\mathbb{C} \setminus (-\infty,0]$ and its inverse $\exp^{-1} = \ln$ is a
principal value of the logarithmic function. Because $\textstyle z \pm
\sqrt{z^2 - 1} \neq 0$ for every $z \in \mathbb{C}$, $\textstyle z \pm
\sqrt{z^2 - 1} < 0$ only for $z \in (-\infty,-1]$ and $\textstyle z \pm
\sqrt{z^2 - 1} = \left(z \mp \sqrt{z^2 - 1}\right)^{-1}$
the function
\begin{equation}
f(z) = \frac{\ln(z \pm \sqrt{z^2 - 1})}{\pm\sqrt{z^2 - 1}}
\end{equation}
is well defined for $z \in \mathbb{C} \setminus (-\infty,-1]$ and independent on
the choice of the square root $\textstyle \pm\sqrt{z^2 - 1}$ of $z^2 - 1$ used
in its definition. If we define for $A \in \mathcal{U}$
\begin{equation}
\exp^{-1}(A) = f(\tfrac{1}{2}\Tr A)\left(A - \tfrac{1}{2}(\Tr A) I\right),
\end{equation}
then it is easy to check that $\exp^{-1}$ given by the above formula is indeed
an inverse function to $\exp \colon \mathcal{O} \to \mathcal{U}$. This proves
bijectivity of $\exp$.

Note, that $\exp \colon \mathfrak{sl}(2,\mathbb{C}) \to SL(2,\mathbb{C})$ is not
surjective \cite{Duistermaat:2000}. In fact the set of matrices not in the image
of $\exp$ is equal
\begin{equation}
SL(2,\mathbb{C}) \setminus \exp(\mathfrak{sl}(2,\mathbb{C}))
    = \{A \in SL(2,\mathbb{C}) \mid \Tr A = -2,\ A \neq -I\}.
\end{equation}

Remains to check that $\exp^{-1} \colon \mathcal{U} \to \mathcal{O}$ is smooth.
Similarly as in Section~\ref{subsec:6.2} this will be the case if for all
$X \in \mathcal{O}$ eigenvalues $\lambda_j$ of $\ad_X$ will satisfy:
$\lambda_j \neq 2k\pi i$, $k = \pm 1,\pm 2,\dotsc$. Observe that if $\pm\lambda$
are eigenvalues of $X$, then eigenvalues of $\ad_X$ are equal
$0,-2\lambda,2\lambda$. Since $\abs{\im\lambda} < \pi$ it follows that
$\exp^{-1}$ is smooth.

Now we will show that $SL(2,\mathbb{C}) \setminus \mathcal{U}$ is of measure
zero. Let $\Sigma = \{A \in SL(2,\mathbb{C}) \mid \Tr A \in \mathbb{R}\}$.
Define function $\Phi \colon SL(2,\mathbb{C}) \to \mathbb{R}$ by the formula
\begin{equation}
\Phi(A) = \im\Tr A.
\end{equation}
This function is smooth and its differential $\dd\Phi(A)$ is nonzero for every
$A \in SL(2,\mathbb{C})$. Thus the level set $\Phi^{-1}(0) = \Sigma$ will be a
submanifold of $SL(2,\mathbb{C})$ of dimension $\dim\Sigma =
\dim SL(2,\mathbb{C}) - 1$ \cite[Corollary~5.14]{Lee:2003}. Therefore, $\Sigma$
will be of measure zero \cite[Corollary~6.12]{Lee:2003} and consequently
$SL(2,\mathbb{C}) \setminus \mathcal{U} \subset \Sigma$ as well.

The special linear group $SL(2,\mathbb{C})$ is an example of a non-compact
simply connected Lie group. Note, however, that its subgroup $SL(2,\mathbb{R})$
is not weakly exponential. The sets $\mathcal{O}$ and $\mathcal{U}$ can be
defined in an analogical way and $\exp \colon \mathcal{O} \to \mathcal{U}$ will
be a diffeomorphism, but $SL(2,\mathbb{R}) \setminus \mathcal{U}$ will not have
measure zero. In fact $\{A \in SL(2,\mathbb{R}) \mid \Tr A < -2\} \subset
SL(2,\mathbb{R}) \setminus \mathcal{U}$ is a non-empty open subset and as such
has non-zero measure. On the other hand $SU(2)$ is also a subgroup of
$SL(2,\mathbb{C})$, which as we will see, is weakly exponential.

\subsection{The group $SU(2)$}
\label{subsec:6.4}
Another example of a weakly exponential Lie group is the special unitary group
$SU(2)$. It is the group of $2 \times 2$ complex unitary matrices with
determinant equal 1 ($U^\dagger U = I$, $\det U = 1$). The Lie algebra
$\mathfrak{su}(2)$ of $SU(2)$ consists of all complex skew-Hermitian
$2 \times 2$ matrices with trace zero ($X^\dagger = -X$, $\Tr X = 0$). On
$\mathfrak{su}(2)$ we can introduce a norm by the following formula:
$\norm{X} = 2\sqrt{\det X}$. Note, that $\norm{\sdot}$ is compatible with the
Lie bracket on $\mathfrak{su}(2)$: $\norm{[X,Y]} \leq \norm{X}\norm{Y}$.

The exponential map $\exp$ maps an open ball $\mathcal{O} =
\{X \in \mathfrak{su}(2) \mid \norm{X} < 2\pi\}$ diffeomorphically onto
$\mathcal{U} = SU(2) \setminus \{-I\}$, see \cite{Duistermaat:2000}. Indeed,
eigenvalues of $X \in \mathfrak{su}(2)$ are equal $\pm\tfrac{1}{2}i\norm{X}$ and
for $U \in SU(2)$ we have that $\Tr U \in \mathbb{R}$ and $\abs{\Tr U} \leq 2$.
Thus the sets $\mathcal{O}$ and $\mathcal{U}$ corresponding to the group
$SL(2,\mathbb{C})$ and intersected with $\mathfrak{su}(2)$ and $SU(2)$,
respectively, are equal to the sets $\mathcal{O}$ and $\mathcal{U}$ defined
above. The exponential map $\exp \colon \mathcal{O} \to \mathcal{U}$ will be
then a diffeomorphism. Clearly $SU(2) \setminus \mathcal{U} = \{-I\}$ is of
measure zero. The special unitary group $SU(2)$ is an example of a compact
simply connected Lie group.

\section{Conclusions and final remarks}
\label{sec:7}
In the paper was presented a complete theory of quantization of a classical
Hamiltonian system whose configuration space is in the form of a Lie group. The
received theory can be now used to quantize particular systems, like a rigid
body. The configuration space of this system is the rotation group $SO(3)$
representing rotational degrees of freedom. The translational degrees of freedom
can be included by taking as the configuration space the semi-direct product of
$\mathbb{R}^3$ with $SO(3)$, which is equal to a special Euclidean group
$E^+(3)$.

The use of deformation quantization approach to quantize a rigid body might help
in developing in quantum mechanics techniques from a classical theory.
Particular examples of rigid bodies are tops describing a precession of a body
under the influence of gravity. Especially interesting are Euler, Lagrange, and
Kovalevskaya tops, which are integrable Hamiltonian systems. It would be
interesting to use the developed quantization theory to construct quantum
versions of these systems together with a theory of quantum integrability.

Worth noting are papers \cite{Lazard:1966,Dokovic:1997} where authors determined
for the exponential function $\exp \colon \mathfrak{g} \to G$ maximal open
domains in $\mathfrak{g}$ on which $\exp$ is injective regardless of the
structure of $\mathfrak{g}$ or $G$. In \cite{Lazard:1966} authors showed that if
we choose a norm on $\mathfrak{g}$ such that $\norm{[X,Y]} \leq
\norm{X}\norm{Y}$, then $\exp$ will be injective on an open ball
$B_\pi = \{X \in \mathfrak{g} \mid \norm{X} < \pi\}$ of radius $\pi$
(and if $G$ is simply connected, then $\exp$ will be injective on an open ball
$B_{2\pi}$ of radius $2\pi$). The presented examples of the groups $SO(3)$ and
$SU(2)$ agree with this result. In \cite{Dokovic:1997} authors introduced a
function $\sigma \colon \mathfrak{g} \to \mathbb{R}^+$ defined by the formula
\begin{equation}
\sigma(X) = \max\{\abs{\im\lambda} \mid \lambda \in \Spec(\ad_X)\}.
\end{equation}
Then they showed that $\exp$ is a diffeomorphism of $\mathcal{O} =
\{X \in \mathfrak{g} \mid \sigma(X) < \pi\}$ onto $\mathcal{U} =
\exp(\mathcal{O})$. This result agrees with the presented examples of the groups
$SO(3)$, $SU(2)$, and $SL(2,\mathbb{C})$. In fact for the groups $SU(2)$ and
$SL(2,\mathbb{C})$ the exponential function is diffeomorphic on a larger set
equal $\mathcal{O} = \{X \in \mathfrak{g} \mid \sigma(X) < 2\pi\}$, which is
caused by the fact that these groups are simply connected.

In the paper we introduced a topology on the space $\mathcal{F}(M)$ in the
case of a compact group $G$. This was done by establishing an isomorphism
between $\mathcal{F}(M)$ and $C^\infty(G \times G)$ and with its help
transferring a natural Fr\'echet space topology on $C^\infty(G \times G)$ onto
$\mathcal{F}(M)$. To extend this construction in a meaningful way to non-compact
groups, so that the \nbr{\star}product would remain continuous and the results
of Section~\ref{subsec:3.4} would still hold, we would have to extend the space
$\mathcal{F}(M)$ by extending the corresponding space $C^\infty_c(G \times G)$
of integral kernels to some bigger space with nicer properties. One of the
candidates worth investigating is a Harish-Chandra's Schwartz space of functions
on $G \times G$ whose derivatives are rapidly decreasing \cite{Wallach:1988}.

The presented approach to quantization can be reformulated and applied to
systems which configuration spaces are in the form of a Riemann-Cartan manifold,
i.e.\ a manifold endowed with a metric tensor and a metric affine connection.
An analog of the weakly exponential Lie group will be a weakly geodesically
simply connected Riemann-Cartan manifold, i.e.\ a Riemann-Cartan manifold
$\mathcal{Q}$ such that for every point $q \in \mathcal{Q}$ there exists a
neighborhood $\mathcal{U}_q \subset \mathcal{Q}$ of $q$ for which $\mathcal{Q}
\setminus \mathcal{U}_q$ is of measure zero and every point in $\mathcal{U}_q$
can be connected with $q$ by a unique geodesics. Then the exponential map
$\exp_q$ will map diffeomorphically some open neighborhood $\mathcal{O}_q
\subset \T_q \mathcal{Q}$ of 0 onto $\mathcal{U}_q$. In fact a semi-simple Lie
group is an example of a Riemann-Cartan manifold which metric tensor is given by
a Killing-Cartan form and which affine connection is such that the exponential
map at unit element is equal to the exponential map on the Lie group. Such
Riemann-Cartan manifold has vanishing curvature tensor and non-vanishing torsion
tensor in the non-Abelian case.


\begin{thebibliography}{35}
\providecommand{\natexlab}[1]{#1}
\providecommand{\url}[1]{\texttt{#1}}
\providecommand{\urlprefix}{URL }
\providecommand{\selectlanguage}[1]{\relax}
\providecommand{\eprint}[2][]{\url{#2}}

\bibitem[{Bayen et~al.(1975--1977)Bayen, Flato, Fr{\o}nsdal, Lichnerowicz, and
  Sternheimer}]{Bayen:1975-1977}
F.~Bayen, M.~Flato, C.~Fr{\o}nsdal, A.~Lichnerowicz, and D.~Sternheimer.
\newblock \emph{{Quantum mechanics as a deformation of classical mechanics}}.
\newblock Lett. Math. Phys. \textbf{1}(6), (1975--1977) pp. 521--530

\bibitem[{Bayen et~al.(1978{\natexlab{a}})Bayen, Flato, Fr{\o}nsdal,
  Lichnerowicz, and Sternheimer}]{Bayen:1978a}
F.~Bayen, M.~Flato, C.~Fr{\o}nsdal, A.~Lichnerowicz, and D.~Sternheimer.
\newblock \emph{{Deformation theory and quantization. I. Deformations of
  symplectic structures}}.
\newblock Ann. Phys. \textbf{111}(1), (1978{\natexlab{a}}) pp. 61--110

\bibitem[{Bayen et~al.(1978{\natexlab{b}})Bayen, Flato, Fr{\o}nsdal,
  Lichnerowicz, and Sternheimer}]{Bayen:1978b}
F.~Bayen, M.~Flato, C.~Fr{\o}nsdal, A.~Lichnerowicz, and D.~Sternheimer.
\newblock \emph{{Deformation theory and quantization. II. Physical
  applications}}.
\newblock Ann. Phys. \textbf{111}(1), (1978{\natexlab{b}}) pp. 111--151

\bibitem[{Zachos et~al.(2005)Zachos, Fairlie, and Curtright}]{Zachos:2005}
C.~K. Zachos, D.~B. Fairlie, and T.~L. Curtright (Editors).
\newblock \emph{{Quantum Mechanics in Phase Space. An Overview with Selected
  Papers}}, volume~34 of \emph{{World Scientific Series in 20th Century
  Physics}}.
\newblock World Scientific Publishing Co., New Jersey, Hackensack (2005)

\bibitem[{de~Gosson(2006)}]{Gosson:2006}
M.~de~Gosson.
\newblock \emph{{Symplectic Geometry and Quantum Mechanics}}, volume 166 of
  \emph{{Operator Theory: Advances and Applications}}.
\newblock Birkh{\"a}user, Basel (2006)

\bibitem[{B{\l}aszak and Doma{\'n}ski(2012)}]{Blaszak:2012}
M.~B{\l}aszak and Z.~Doma{\'n}ski.
\newblock \emph{{Phase space quantum mechanics}}.
\newblock Ann. Phys. \textbf{327}(2), (2012) pp. 167--211.
\newblock \eprint{arXiv:1009.0150 [math-ph]}

\bibitem[{Underhill(1978)}]{Underhill:1978}
J.~Underhill.
\newblock \emph{{Quantization on a manifold with connection}}.
\newblock J. Math. Phys. \textbf{19}(9), (1978) pp. 1932--1935

\bibitem[{Liu and Quian(1992)}]{Liu:1992}
Z.~J. Liu and M.~Quian.
\newblock \emph{{Gauge invariant quantization on Riemannian manifolds}}.
\newblock Trans. Amer. Math. Soc. \textbf{331}(1), (1992) pp. 321--333

\bibitem[{Pflaum(1998)}]{Pflaum:1998}
M.~J. Pflaum.
\newblock \emph{{A deformation-theoretical approach to Weyl quantization on
  Riemannian manifolds}}.
\newblock Lett. Math. Phys. \textbf{45}(4), (1998) pp. 277--294

\bibitem[{Pflaum(1999)}]{Pflaum:1999}
M.~J. Pflaum.
\newblock \emph{{Deformation quantization on cotangent bundles}}.
\newblock Rep. Math. Phys. \textbf{43}(1--2), (1999) pp. 291--297

\bibitem[{Bordemann et~al.(1998)Bordemann, Neumaier, and
  Waldmann}]{Bordemann.Neumaier.Waldmann:1998}
M.~Bordemann, N.~Neumaier, and S.~Waldmann.
\newblock \emph{{Homogeneous Fedosov star products on cotangent bundles I: Weyl
  and standard ordering with differential operator representation}}.
\newblock Commun. Math. Phys. \textbf{198}(2), (1998) pp. 363--396.
\newblock \eprint{arXiv:q-alg/9707030}

\bibitem[{Bordemann et~al.(1999)Bordemann, Neumaier, and
  Waldmann}]{Bordemann.Neumaier.Waldmann:1999}
M.~Bordemann, N.~Neumaier, and S.~Waldmann.
\newblock \emph{{Homogeneous Fedosov star products on cotangent bundles II: GNS
  representations, the WKB expansion, traces, and applications}}.
\newblock J. Geom. Phys. \textbf{29}(3), (1999) pp. 199--234.
\newblock \eprint{arXiv:q-alg/9711016}

\bibitem[{Karasev and Osborn(2005)}]{Karasev:2005}
M.~V. Karasev and T.~A. Osborn.
\newblock \emph{{Cotangent bundle quantization: entangling of metric and
  magnetic field}}.
\newblock J. Phys. A \textbf{38}, (2005) pp. 8549--8578

\bibitem[{Rieffel(1989{\natexlab{a}})}]{Rieffel:1989a}
M.~A. Rieffel.
\newblock \emph{{Continuous fields of $C^*$-algebras coming from group cocycles
  and actions}}.
\newblock Math. Ann. \textbf{283}, (1989{\natexlab{a}}) pp. 631--643

\bibitem[{Rieffel(1989{\natexlab{b}})}]{Rieffel:1989b}
M.~A. Rieffel.
\newblock \emph{{Deformation quantization for Heisenberg manifolds}}.
\newblock Commun. Math. Phys. \textbf{122}(4), (1989{\natexlab{b}}) pp.
  531--562

\bibitem[{Rieffel(1994)}]{Rieffel:1994}
M.~Rieffel.
\newblock \emph{{Quantization and $C^*$-algebras}}.
\newblock Contemp. Math. \textbf{167}, (1994) pp. 67--97

\bibitem[{Natsume(2001)}]{Natsume:2001}
T.~Natsume.
\newblock \emph{{$C^*$-algebraic deformation quantization and the index
  theorem}}.
\newblock Math. Phys. Stud. \textbf{23}, (2001) pp. 142--150

\bibitem[{Natsume et~al.(2003)Natsume, Nest, and Ingo}]{Natsume:2003}
T.~Natsume, R.~Nest, and P.~Ingo.
\newblock \emph{{Strict quantizations of symplectic manifolds}}.
\newblock Lett. Math. Phys. \textbf{66}, (2003) pp. 73--89

\bibitem[{Cahen and Gutt(1982)}]{Cahen:1982}
M.~Cahen and S.~Gutt.
\newblock \emph{{Regular $*$-representations of Lie algebras}}.
\newblock Lett. Math. Phys. \textbf{6}(5), (1982) pp. 395--404

\bibitem[{Gutt(1983)}]{Gutt:1983}
S.~Gutt.
\newblock \emph{{An explicit $*$-product on the cotangent bundle of a Lie
  group}}.
\newblock Lett. Math. Phys. \textbf{7}, (1983) pp. 249--258

\bibitem[{Tounsi(2003)}]{Tounsi:2003}
K.~Tounsi.
\newblock \emph{{Integral formulae and Kontsevich star products on the
  cotangent bundle of a Lie group}}.
\newblock Ital. J. Pure Appl. Math. \textbf{14}, (2003) pp. 137--147

\bibitem[{Tounsi(2010)}]{Tounsi:2010}
K.~Tounsi.
\newblock \emph{{Homogeneous star products and closed integral formulas}}.
\newblock Tohoku Math. J. \textbf{62}, (2010) pp. 559--573

\bibitem[{Arnal et~al.(1999)Arnal, Ben~Amar, and Masmoudi}]{Arnal:1999}
D.~Arnal, N.~Ben~Amar, and M.~Masmoudi.
\newblock \emph{{Cohomology of good graphs and Kontsevich linear star
  products}}.
\newblock Lett. Math. Phys. \textbf{48}, (1999) pp. 291--306

\bibitem[{Dito(1999)}]{Dito:1999}
G.~Dito.
\newblock \emph{{Kontsevich star product on the dual of a Lie algebra}}.
\newblock Lett. Math. Phys. \textbf{48}, (1999) pp. 307--322

\bibitem[{Kathotia(2000)}]{Kathotia:2000}
V.~Kathotia.
\newblock \emph{{Kontsevich's universal formula for deformation quantization
  and the Campbell-Baker-Hausdorff formula}}.
\newblock Int. J. Math. \textbf{11}(04), (2000) pp. 523--551

\bibitem[{Ben~Amar(2003)}]{BenAmar:2003}
N.~Ben~Amar.
\newblock \emph{{K-star products on dual of Lie algebras}}.
\newblock J. Lie Theory \textbf{13}(2), (2003) pp. 329--357

\bibitem[{Gutt(2011)}]{Gutt:2011}
S.~Gutt.
\newblock \emph{{Deformation quantization of Poisson manifolds}}.
\newblock Geom. Topol. Mon. \textbf{17}, (2011) pp. 171--220

\bibitem[{Landsman(1993{\natexlab{a}})}]{Landsman:1993a}
N.~P. Landsman.
\newblock \emph{{Deformations of algebras of observables and the classical
  limit of quantum mechanics}}.
\newblock Rev. Math. Phys. \textbf{05}(04), (1993{\natexlab{a}}) pp. 775--806

\bibitem[{Duistermaat and Kolk(2000)}]{Duistermaat:2000}
J.~J. Duistermaat and J.~A.~C. Kolk.
\newblock \emph{{Lie Groups}}.
\newblock Universitext. Springer, Berlin, Heidelberg, first edition (2000)

\bibitem[{Lee(2003)}]{Lee:2003}
J.~M. Lee.
\newblock \emph{{Introduction to Smooth Manifolds}}.
\newblock Springer, New York (2003)

\bibitem[{Dokovi{\'c} and Hofmann(1997)}]{Dokovic:1997}
D.~{\v Z}. Dokovi{\'c} and K.~H. Hofmann.
\newblock \emph{{The surjectivity question for the exponential function of real
  Lie groups: A status report}}.
\newblock J. Lie Theory \textbf{7}, (1997) pp. 171--199

\bibitem[{Landsman(1993{\natexlab{b}})}]{Landsman:1993b}
N.~P. Landsman.
\newblock \emph{{Strict deformation quantization of a particle in external
  gravitational and Yang-Mills fields}}.
\newblock J. Geom. Phys. \textbf{12}, (1993{\natexlab{b}}) pp. 93--132

\bibitem[{Gracia-Bond{\'\i}a and V{\'a}rilly(1988)}]{Gracia-Bondia:1988}
J.~M. Gracia-Bond{\'\i}a and J.~C. V{\'a}rilly.
\newblock \emph{{Algebras of distributions suitable for phase-space quantum
  mechanics}}.
\newblock J. Math. Phys. \textbf{29}(4), (1988) pp. 869--879

\bibitem[{Lazard and Tits(1966)}]{Lazard:1966}
M.~Lazard and J.~Tits.
\newblock \emph{{Domaines d'injectivit\'e de l'application exponentielle}}.
\newblock Topology \textbf{4}, (1966) pp. 315--322

\bibitem[{Wallach(1988)}]{Wallach:1988}
N.~R. Wallach.
\newblock \emph{{Real Reductive Groups I}}, volume 132 of \emph{Pure and
  Applied Mathematics}.
\newblock Academic Press, San Diego (1988)
\end{thebibliography}

\end{document}